\documentclass{article}
\usepackage{etex}

\usepackage[T1]{fontenc}
\usepackage[utf8]{inputenc}
\usepackage{xspace}
\usepackage[english]{babel}
\usepackage{amsfonts}
\usepackage{amsmath}
\usepackage{amssymb}
\usepackage{aeguill}
\usepackage{amsthm} 

\usepackage{mathrsfs}
\usepackage[usenames,dvipsnames]{xcolor}
\usepackage[unicode,hyperindex,colorlinks=true,citecolor=red,linkcolor=blue]{hyperref}
\usepackage{geometry}
\usepackage{stmaryrd}

\usepackage[nottoc, notlof, notlot]{tocbibind}
\usepackage{todonotes}
\usepackage{comment}

\usepackage{float}
\floatstyle{boxed} 
\restylefloat{figure}

\usepackage{appendix}

\usepackage{MnSymbol}
\usepackage{graphicx}
\usepackage{multicol}
\input{diagxy}
\usepackage{bussproofs}
\usepackage{cmll}

\usepackage[all]{xy}

\newtheorem{defi}{Definition}[section]
\newtheorem{theorem}[defi]{Theorem}
\newtheorem{prop}[defi]{Proposition}

\newtheorem{lemma}[defi]{Lemma}
\newtheorem{coro}[defi]{Corollary}

\newcommand{\bornof}{{bounded}\xspace}
\newcommand{\lct}{\xspace{lctvs}\xspace}
\newcommand{\Lin}{\mathbf{Lin}}
\newcommand{\Smo}{\mathbf{Smooth}}
\newcommand{\Quant}{\mathbf{Quant}}
\newcommand{\linb}{\mathcal{L}}
\newcommand{\nlin}{\linb^n}
\newcommand{\pol}{\mathcal{P}}
\newcommand{\Cin}{\mathcal{C}^{\infty}} 
\newcommand{\series}{S}

\newcommand{\tb}[1]{\tau_b(#1)} 
\newcommand{\Ec}{E^{\times}}
\newcommand{\Ecc}{E^{\times \times}}
\newcommand{\Fc}{F^{\times}}

\newcommand{\CC}{\mathbb{C}}
\newcommand{\RR}{\mathbb{R}}
\newcommand{\W}{\mathcal{W}}
\newcommand{\DD}{\mathbb{D}}
\newcommand{\NN}{\mathbb{N}}
\newcommand{\CCwo}{\mathbb{C}^\ast}
\newcommand{\RRp}{\mathbb{R}^+}
\newcommand{\ie}{i.e.\xspace}
\newcommand{\wlg}{w.l.o.g.\xspace}
\newcommand{\eg}{e.g.\xspace}
\newcommand{\mco}{\xspace{Mackey-complete}\xspace}
\newcommand{\mca}{{Mackey-cauchy}\xspace}
\newcommand{\mctens}{\hat\otimes}
\newcommand{\unit}{1}

\newcommand{\conv}{\ast}

\newcommand{\abs}[1]{\mathopen{|}#1\mathclose{|}}
\newcommand{\cC}{\mathcal C}
\newcommand{\cB}{\mathcal B}

\date{\today}
\author{Marie Kerjean and Christine Tasson}
\title{\mco spaces and power series -- A topological model of Differential Linear Logic}

\begin{document}

\label{firstpage}
\maketitle

\begin{abstract}
In this paper, we have described a denotational model of Intuitionist Linear Logic which is also a differential category. Formulas are interpreted as Mackey-complete topological vector space and linear proofs are interpreted by bounded linear functions. So as to interpret non-linear proofs of Linear Logic, we have used a notion of power series between Mackey-complete spaces, generalizing the notion of entire functions in $\CC$. Finally, we have obtained a quantitative model of Intuitionist Differential Linear Logic, where the syntactic differentiation correspond to the usual one and where the interpretations of proofs satisfy a Taylor expansion decomposition.

\end{abstract}


\tableofcontents

\section{Introduction}

Logic is by nature discrete, and linear logic is not different. The interpretation of linearity in terms of resource consumption still manipulates discrete notions, \ie proofs are seen as operators on multisets of formulas. Many denotational models of linear logic are also discrete, for example, based  on graphs such as coherent spaces~\cite{systemflater}, on games~\cite{ho,ajm}, or on sets and relations which can also be endowed with an additive structure giving rise to vector spaces with bases~\cite{Ehr02,Ehr05}. Furthermore, Ehrhard and Regnier explain in~\cite{ER06diff} and~\cite{ER03} how it is possible to add a differentiation rule to Linear Logic, constructing this way Differential Linear Logic (DiLL). In this work, differentiation is seen as a way to transform a non-linear proof $ f : A \Rightarrow B$ into a linear proof $ Df: A \multimap (A\Rightarrow B)$. In some models such as the relational model, differentiation has a combinatorial interpretation. Continuous models of DiLL, where non-linear proofs are interpreted by differentiable functions, are even more appealing, as the synctactical differentiation corresponds to the mathematical one. In~\cite{Ehr02} and~\cite{Ehr05}, for instance, non-linear proofs are interpreted by power series between  Köthe spaces and Finiteness spaces respectively, that are sequence spaces. One could even ask for an interpretation of the differentiation rule in more general spaces.

\paragraph*{Bornologies.}
The search for topological models of Linear Logic relies on some fundamental mathematical issues. Indeed, having a cartesian closed category of topological spaces is not straightforward. Several answers exists (see~\cite{escardo} for a past account), and among them is the definition of convenient spaces and  smooth functions by Fr\"olicher, Kriegl and Michor in~\cite{FroKri} and~\cite{KriMi}. Those are the smooth functions used in~\cite{BET10} for modelling DiLL. Moreover, as explained by Girard in the introduction of~\cite{Gir96}, if the proofs are interpreted by continuous functions, then, notably, the interpretations of the proofs of $A, A \Rightarrow B \vdash B$ and  of $A \vdash (A \Rightarrow B) \Rightarrow B$ are also continuous. That is, $ x, f \mapsto f(x) $ and $ x \mapsto (\delta_x : f \mapsto f(x)) $ must be continuous. This would be the case if linear function spaces bore both a uniform convergence and a pointwise convergence topology. We believe that this is solved by the use of bounded sets, \ie by using the advantages of the theory of bornologies (see~\cite{Hogbe} for an overview of this theory). Indeed, the Banach-Steinhauss theorem says that between Banach spaces, the topology of uniform convergence on bounded sets and the pointwise convergence topology on a space of linear functions give rise to the same bounded sets. This theorem is generalized in~\cite{KriMi}, where the authors use Mackey-complete spaces (complete spaces for a specific version of Cauchy sequences) and \bornof linear maps (linear maps preserving bounded sets). 
This observation was exploited in~\cite{FroKri} and~\cite{KriMi} where \bornof linear functions replace continuous ones.

\paragraph*{Quantitative semantics.}
Introduced by Girard in~\cite{Gir88}, quantitative semantics refine the analogy between linear functions and linear programs (consuming exactly one times its input). Indeed, programs consuming exactly $n$-times their resources are seen as monomials of degree $n$. General programs are seen as the disjunction of their executions consuming $n$-times their resources. Mathematically, this means that non-linear programs are interpreted by potentially infinite sums of monomials, that are power series. This analogy can be found in many denotational models of variant of Linear Logic such as Fock spaces~\cite{Fock94} Köthe spaces~\cite{Ehr02}, Finiteness spaces~\cite{Ehr05}, Probabilistic Coherent spaces~\cite{danosehrhard}, or, in a more categorical setting, in analytic functors~\cite{2aaf} and generalised species~\cite{Gspecies}.

\paragraph*{\mco spaces and Power series.}
In this article, we have brought to light a model of Intuitionist Differential Linear Logic, whose objects are \emph{locally convex topological vector spaces} that are \mco (see Definition~\ref{def:mackey}).  The ingredients of the model have been chosen with care so that they correspond cleanly to the constructions of Differential Linear Logic: for instance, vector spaces are used to interpret linearity, and topology to interpret differentiation. 

We use the notion of bounded set when we ask \emph{linear functions} not to be continuous but \emph{\bornof}, that is to send a bounded set on a bounded set. The two notions are closely related, but distinct. As a consequence, the interpretation of the \emph{negation} is based on the bounded dual and not on the usual continuous dual. 

The \emph{multiplicative conjunction} $ \otimes$ of Linear Logic is interpreted by the \emph{\bornof} tensor product of topological vector spaces which has then to be Mackey-completed.

The \emph{additive conjunction} $\with$ and \emph{disjunction} $\oplus$ are interpreted respectively by the cartesian product and the coproduct in the category of \mco spaces and \bornof linear functions. Finite products and coproducts coincide, so that the category is equipped with finite biproducts. Notice that if we wanted to ensure that the bounded dual of infinite products are coproducts (the reverse comes automatically), we would need 
to work with spaces whose cardinals are not strongly inaccessible \cite[13.5.4]{Jar81}. This assumption is not restrictive as it is always possible to construct a model of ZFC with non accessible cardinals.

Non-linear proofs of DiLL are interpreted by \emph{power series}, that are sums of \bornof $n$-monomials. In order to work with these functions, we must make use of the theory of holomorphic maps developed in the second chapter of~\cite{KriMi}. This is made possible since the spaces we consider are in particular Mackey-complete. We have proven that the category of \mco spaces and power series is cartesian closed, by generalizing the Fubini theorem over distributions $\series(E\times F,\CC)\simeq \series(E,\series(F,\CC))$ and by using interchange of converging summations in $\CC$. The exponential modality is interpreted as a \mco subspace of  the \bornof dual of the space of scalar power series. Indeed, any space can be embedded in its \bornof bidual $!E\subset (!E)^{\times\times}=(!E\multimap \bot)^\times$ and using the key decomposition $!E \multimap \bot \simeq E \Rightarrow \bot=\series(E,\CC)$ of Linear Logic gives us that $!E\subseteq \series(E,\CC)^\times$. Finally, because we are working with topological vector spaces, the interpretation of the co-dereliction rule of DiLL is the operator taking the directed derivative at $0$ of a function. 



\paragraph*{Related works.}

Our model follows a long history of models establishing connections between analyticity and computability. 

Fock spaces~\cite{Fock94} and Coherent Banach spaces~\cite{Gir96} were the first step towards a continuous semantics of Linear Logic. More precisely, Fock spaces are Banach spaces and Coherent Banach spaces are dual pairs  of Banach spaces (see~\cite[Chap. 8]{Jar81} for an overview of the theory of dual pairs). In Fock spaces, linear programs are interpreted as contractive bounded linear maps and general programs as holomorphic or analytic functions. Similarly, in Coherent Banach spaces, linear programs are interpreted as continuous linear functions and general programs as bounded analytic functions defined on the open unit ball. Yet, it is easy to guess that neither Fock spaces nor Coherent Banach spaces are completely a model of the entire linear logic, but they are a model of a linear exponential, that is of weakening. However, it is remarkable that both works were already using bounded sets (\eg bounded linear forms and continuous linear forms correspond on Banach spaces) and we take advantage of replacing Banach spaces norms by bornologies.

With Köthe spaces~\cite{Ehr02} and then Finiteness spaces~\cite{Ehr05}, Ehrhard designed two continuous semantics of Linear Logic. The objects of the two models are sequence spaces equipped with a structure of topological vector spaces. Köthe spaces are locally convex spaces over the usual real or complex fields, whereas Finiteness spaces are endowed with a linearised topology over a field (potentially of reals or complexes) endowed with discrete topology. The linear proofs are interpreted by continuous linear functions and the non-linear ones by analytic mappings. Even if the interpretation of linear logic formula enjoys an intrinsic characterization, these models are related to the relational semantics. Indeed, a Linear Logic formula is interpreted by a space of sequences whose indices constitute its relational interpretation. Furthermore, the interpretation of a proof is a sequence whose support (the indices of non zero coefficients) is its relational interpretation. Although interpretations of formulas may differ, proofs are identically interpreted in Köthe or Finiteness models (and in the model presented in the present article). The main difference between our model and these Köthe or Finiteness spaces models is precisely that \mco spaces do not have to be sequence spaces. They digress from the discrete setting of the relational model. Since Köthe spaces and \mco spaces are both endowed with locally convex topology, one could think that the first are a special case of the last. However, the function spaces are endowed with the compact open topology for Köthe spaces and with the bounded open topology for the \mco spaces. In particular, the dual $E_X^{\bot}$ of a K\"othe Space $E_X$ is isomorphic to the topological dual of $E_X$, which is in general a strict subset of the bornological dual (all K\"othe spaces are not bornological). It raises an interesting question about whether a description with bounded subsets would help having an intrinsic description of K\"othe spaces. On the contrary, although Finiteness spaces do not have the same kind of topology, their use of bounded sets is central and our model borrows a lot of Finiteness spaces constructions.

The present work is thought as a restriction of Convenient spaces~\cite{BET10}, that is \mco spaces and smooth maps. In this model of Intuitionist Linear Logic, which is a differential category, non-linear proofs are interpreted with some specific smooth maps. No references are made to a discrete setting, but as in Finiteness spaces, the topology and the bornology are dually related. Although this bornological condition facilitates
 the proofs, it is not necessary to interpret Intuitionnist Linear Logic. 
Thus, in our model, we release the bornological condition on the topology

Remember that in many Quantitative models of Linear Logic, as in Normal functors~\cite{Gir88,2aaf}, Fock spaces~\cite{Fock94} or Finiteness spaces~\cite{Ehr05,Ehr07} non-linear proofs are interpreted as analytic functions.
In our model, we refine smooth maps into analytic ones. On the way, we consider topological vector spaces over $\CC$ to be able to handle holomorphic functions. This is another difference with Convenient Vector spaces as presented in~\cite{BET10}.


\paragraph{Content of the paper.}

We begin the paper by laying down the bornological setting (Subsection~\ref{subsec:borno}) and by defining the central notion of \mco spaces (Subsection~\ref{subsec:mco}). Then, in Section~\ref{sec:lin}, we begin the definition of the model by the linear category of \mco spaces and \bornof linear maps that is cartesian and symetric monoidal closed. This linear part is the base of the present work, but also of the model of \mco spaces and smooth functions introduced in~\cite{BET10}. We have given an overview of this work in a slightly different setting in Section~\ref{sec:smooth} in order to properly describe the landscape of our work. Finally, in Section~\ref{sec:pws}, we introduce the power series, their definition and properties that are useful in demonstrating that \mco spaces and power series constitute a quantitative model of Intuitionistic Differential Linear Logic.

\paragraph{Acknowledgements.}
The authors thank Rick Blute and Thomas Ehrhard for the inspiration and the lively conversations.

\section{Preliminaries}

\subsection{Topologies and Bornologies}\label{subsec:borno}

Let us first set the topological scene. We will handle \emph{complex} topological vector spaces. We denote by $\CC$ the field of complex and by $\CCwo=\CC\setminus \{0\}$.

The monoidal structure could have been described either with complex or real vector spaces. However, in section~\ref{sec:pws}, we work with power series and make an extensive use of their holomorphic properties. 

More precisely, we will work with \emph{locally convex separated topological vector spaces}  (see~\cite{Jar81} I.2.1) and refer to them as \emph{lctvs}. From now on, $E$ and $F$ denote \lct.  A set $C$ in a  $ \CC$-vector space is said to be \emph{absolutely convex} when for all $x, y \in C$, for all $\lambda , \mu  \in \CC$, if $|\lambda|+|\mu| < 1$ then $\lambda x + \mu y \in C$. By definition, the topology of an \lct is generated by a basis of neighbourhood of $0$ made of absolutely convex subsets. We will use that if $C$ is an absolutely convex subset of an \lct, then $\bar{C} \subset 3 C$, and  $\lambda C + \mu C \subset (\lambda + \mu) C$ for all $\lambda , \mu\in \CC $.

\paragraph{Bounded sets.} We will also work with \emph{bornologies}, that is collections of bounded sets with specific closure properties. A \emph{subset} $b$ of an \lct is \emph{bounded} when it is absorbed by every $0$-neighbourhood $U$, that is there is $\lambda \in \CC $ such that $b \subseteq \lambda U$. A \emph{disk} is a bounded absolutely convex set.
A \emph{function} is \emph{\bornof} when it sends a bounded set of its domain on a bounded set of its codomain.  Two spaces are \emph{\bornof equivalent}, noted $E\simeq F$, when there is a bijection $\phi:E\rightarrow F$ such that $\phi$ and $\phi^{-1}$ are both linear and  bounded.

 Let us denote $E'$ the space of linear \emph{continuous} forms on $E$, $\Ec$ the space of linear \emph{\bornof} forms on $E$, and $E^{\star}$ the space of linear forms on $E$. 
Remark that any linear continuous function is bounded and so $E' \subset \Ec\subset E^\star$. 

\paragraph{The Mackey-Arens Theorem.} It is a fundamental theorem for the theory of bornologies. It states that bounded subsets can be characterized as the one that are sent to a bounded ball by any continuous linear form. We state it for bounded linear forms.

\begin{lemma}
\label{scalbound}
A subset $b \subset E$ is bounded if and only if it is scalarly bounded, that is: 
$$\forall \ell \in \Ec,\ \exists \lambda>0,\ |\ell(b)|<\lambda.$$
\end{lemma}

\begin{proof}
By definition of the bounded linear forms, the image of a bounded set is bounded. For the reverse implication, we use  the Mackey-Arens theorem (see for example~\cite[IV.3.2]{Sch71}). Indeed, since for any $l\in E'$, $l\in \Ec$, we have $l(b)$ is bounded, and so $b$ is bounded.
\end{proof}

\paragraph{The Hahn-Banach Theorem.}
Usually, the Hahn-Banach separation theorem is stated for continuous linear forms (see~\cite[proposition 7.2.2.a]{Jar81}). We adapt it to bounded linear forms as $E'\subseteq\Ec$. The principal flaw to the theory of vectorial spaces and bornologies is that there is no version of the Hahn-Banach extension theorem for bounded linear maps~\cite{Hog70} .

\begin{prop}
\label{HBsep}
Let $C$ be a closed convex subset of $E$. If $x \in E\setminus C$, then there is $ \ell \in E'\subset\Ec$ such that $|\ell(x)|= 1 $ and for all $y \in C $ $|\ell(y)| =0 $.
\end{prop}

\paragraph{Bornivorous subsets.} We introduced bounded sets as a definition depending of the topology. It is also possible to define $0$-neighbourhood from a bornology. 
\begin{defi}
A \emph{bornivorous} is a subset $U\subseteq E$ absorbing any bounded up to dilatation:
$\forall b\subset E\text{ bounded},\ \exists\lambda \in \RRp,\ \lambda b \subseteq U.$

The \emph{bornological topology} $\tau_b$ of $E$ is the topology generated by the bornivorous disks of $E$.  
\end{defi}

Note that any neighbourhood of $0$ in the topology of $E$ is bornivorous, but the converse is false, \ie the bornivorous topology $\tau_b$ is finer than the topology of $E$. The point of the bornologification of an \lct is precisely to enrich $E$ with all the bornivourous subsets as $0$-neighbourhood, so that we get better relations between continuity and boundedness (see~\cite[13.1]{Jar81} for details on this notion). 

\begin{prop}\label{prop:topoborno}
\begin{enumerate}
\item The bounded sets of $E$ and $\tb E$ are the same.
\item A linear function $f : E \rightarrow F$ is bounded if and only if $ f: \tb E \rightarrow F$ is continuous.
\end{enumerate}
\end{prop}
\begin{proof}
  The first item stems from definition handling.
  For the second one,
  if $f: \tb E \rightarrow F$ is continuous, it is bounded and because
  $E$ and $\tb E$ bears the same bounded sets $f : E \rightarrow F$ is  bounded. Conversely, suppose that $f : E \rightarrow F$ is bounded. Then one see that when $V$ is a $0$-neighbourhood in $F$, $f^{-1}(V)$ is a bornivorous subset of $E$, hence a $0$-neighbourhood in $\tb E$. Thus ${f} : \tb E \rightarrow F$ is continuous.
  \end{proof}

\subsection{Mackey-complete spaces}
\label{subsec:mco}

Mackey-complete spaces are very common spaces in mathematics as Mackey-completeness is a very weak completeness condition. For example, every complete space, quasi-complete, or weakly complete space is Mackey-complete. Mackey-complete spaces are called locally complete spaces in~\cite{Jar81}, or convenient spaces in~\cite{KriMi}. Although it is not a very restraining notion, Mackey-completeness suffices to speak about smoothness between \lct, in the meaning of Kriegl and Michor~\cite{KriMi}. 

\begin{defi}\label{def:mackey}
Consider $E$ an \lct. A \emph{Mackey-Cauchy net} in $E$ is a net $(x_{\gamma})_{\gamma \in \Gamma}$ such that there is a net of scalars $\lambda_{\gamma, \gamma'}$ decreasing towards $0$ and a bounded set $b$ of $E$ such that:
$$ \forall \gamma, \gamma' \in \Gamma, x_{\gamma} - x_{\gamma'} \in \lambda_{\gamma, \gamma'} b.$$
A space where every Mackey-Cauchy net converges is called \emph{Mackey-complete}.  
\end{defi}

Note that a converging Mackey-Cauchy net does in fact \emph{Mackey-converge}, \ie there is a net of scalars $\lambda_{\gamma}$ decreasing towards $0$ such that $x_{\gamma} - \lim_\gamma x_{\gamma} \in \lambda_{\gamma} b$. Note also that a Mackey-converging net is always a converging net, by definition of boundedness in an \lct. 

Notice that the convergence of Mackey-cauchy nets and the convergence of Mackey-cauchy sequences are equivalent (see~\cite[I.2.2]{KriMi}).
Mackey-converging sequences and bounded functions behave particularly well together. Indeed, a bounded function is not continuous in general, so it does not preserve converging sequences but it preserves Mackey-Cauchy nets.

\begin{prop}
\label{bounded_Mackey}
Bounded linear functions preserves Mackey-convergence and Mackey-Cauchy nets.
\end{prop}

There is a nice characterization of the Mackey-completeness, through a decomposition into a collection of Banach spaces. 

\begin{defi}
\label{E_b}
Consider $b$ an absolutely convex and bounded subset of an \lct $E$. We write $E_b$ for the linear span of $b$ in $E$, and we endow it with the Minkowski functional defined as:
$$ p_b (x) = \inf\{ \lambda\in\RR^+\ |\ x \in \lambda  b \}. $$
It is a normed space.
\end{defi}

As a Mackey-Cauchy net is nothing but a Cauchy net in some specific $E_b$, we have : 

\begin{prop}{\cite[I.2.2]{KriMi}}
\label{E_b_mco}
An \lct $E$ is Mackey-complete if and only if for every bounded and absolutely convex subset $b$, $E_b$ is a Banach space.
\end{prop}



Similarly to what happens in the more classical theory of complete spaces, we have a \textit{Mackey-completion} procedure. This one is slightly more intricate than the completion procedure, as it consists of the right  completion of each of the $E_b$.

\begin{prop}{\cite[I.4.29]{KriMi}}
\label{mcompletion}
For every \lct $E$ there is a unique (up to bounded isomorphism) \mco \lct $\tilde{E}$ and a bounded embedding $ \iota : E \rightarrow \tilde{E}$ 
such that for every \mco \lct $F$, for every bounded linear map $f : E \rightarrow F$ there is a unique bounded linear map 
$\tilde{f} : \tilde{E} \rightarrow F $ extending $f$ such that $ f = \tilde{f} \circ \iota$.
\end{prop}

The Mackey-completion procedure can be decomposed in three steps. First, one bornologizes the space $E_{born}$, so that $E_{born}$ bears a topology where the $0$-neighbourhoods are exactly the bornivorous ones. Then one Cauchy-completes this space into a space $\widetilde{E_{born}}$. The Mackey-completion of $E$ is the Mackey-closure of $E$ in $\widetilde{E_{born}}$. 

Beware that the \textit{Mackey-closure} procedure does not behave as simply as the closure procedure. Indeed, the Mackey-closure of a subset $B$ is the smallest Mackey-closed (i.e. closed for Mackey-convergence) set containing $X$. It does not coincide in general with the Mackey-adherence of $X$, that is the set of all limits of Mackey-converging sequences of elements of $X$, see~\cite[I.4.32]{KriMi}.

Let us describe finally a few preservation properties of Mackey-complete spaces.

\begin{prop}{\cite[I.2.15]{KriMi}}\label{prop:stability_mco}
Mackey-completeness is preserved by limits, direct sums, strict inductive limits of sequences of closed embeddings. It is not preserved in general by quotient nor general inductive limits.
\end{prop}

\paragraph{Spaces of bounded maps.} Let us write $\mathcal{B}(E,F)$ for the space of bounded maps from $E$ to $F$ (not necessarily linear), endowed with the topology of uniform convergence on bounded sets of $E$. As in the linear case (see below), bounded sets of $\mathcal{B}(E,F)$ are the \emph{equibounded} ones, that is the sets $B\subset\mathcal{B}(E,F)$ such that for any $b\subset E$ bounded, $B(b)=\{f(x)\ |\ f\in B,\, x\in b\}$ is bounded in $F$. 

\begin{prop}\label{prop:B(E,F)mco}\cite[I.2.15]{KriMi}
Let $E$ and $F$ be \lct. If $F$ is \mco, then so is $\mathcal{B}(E,F)$.
\end{prop}

\begin{proof}
Consider $(f_{\gamma})_{(\gamma \in \Gamma)}$ a \mca net in $\mathcal{B}(E,F)$, \ie there is $(\lambda_{\gamma, \gamma'})\subset \RR$ decreasing towards $0$ and an equibounded ${B}$ in $\mathcal{B}(E,F)$ such that

$$ f_{\gamma} - f_{\gamma'} \in \lambda_{\gamma, \gamma'}{B}.$$ 

For all $ x \in E$, because $B(\{x\})$ is bounded in $F$, $(f_{\gamma} (x))_{\gamma \in \Gamma}$ is also a \mca net. Besides, $F$ is \mco, so each of these \mca nets converges towards $f(x) \in F$. Let us show that $f$ is bounded. Indeed, consider $b$ a closed bounded set in $E$, and $U$ a $0$-neighbourhood in $F$. As $ B$ is equibounded, there is $\lambda \in \CC$ such that $B(b) \subset \lambda U$. Consider $\gamma_0 \in \Gamma$ such that, if $ \gamma , \gamma'\geq \gamma_0$ then $| \lambda_{\gamma, \gamma'} | < \lambda$. Consider $\mu \in \CC$ such that $f_{\gamma_0} (b) \subset \mu U$. Then for all $ \gamma \geq \gamma_0$, $f_\gamma (b) \subseteq \mu U + \lambda U$. Thus $f(b)$ is in $(\lambda + \mu) \bar{U}$, thus in $3(\lambda + \mu) \bar{U}$. We proved that $f(b)$ is a bounded set, and so $f$ is \bornof.
\end{proof}

\section{A symmetric monoidal closed and cartesian category}
\label{sec:lin}
Let us write $\Lin$ for the category whose objects are \mco spaces, and whose morphisms are linear bounded maps. In this setting, the additives are interpreted by the product and the coproduct, while the multiplicative connectives are interpreted by the tensor product and its dual. The only tricky point is to find a good tensor product in our category : this is possible thanks to the Mackey-completion procedure.

\subsection{The (co)cartesian structure}

\paragraph*{Topological products and coproducts.}
The cartesian product of a countable family of \mco spaces is \mco when endowed with the product topology \cite[I.2.15]{KriMi}. Then, a subset is bounded if and only if it is bounded in each direction. The terminal object $\top$ is the $\{0\}$ \mco space.

\begin{defi}
The coproduct of a countable family of \mco spaces is \mco when endowed with the coproduct topology. The coproduct topology is the finest topology on $\bigoplus_i E_i$ for which the injections $ E_i \rightarrow \bigoplus_i E_i$ are continuous.  
\end{defi}

Then, $ B \subset \bigoplus_i E_i$ is bounded if and only if $\{ i\ |\ \exists x \in B \cap E_i \}$ is finite and if for every $i$, $B \cup E_i $ is bounded in $E_i$. The $\{0\}$ vector space is also the unit $0$ of the coproduct.

Notice that in the finite case, the product and the coproduct coincide algebraically and topologically. In the infinite case, the distinction between product and coproduct corresponds to the disctinction between  the space of complex sequences $\CC^{\NN} =  \prod_{n \in \NN} \CC$ and the space of complex finite sequences $\CC^{(\NN)}=\bigoplus_{n\in\NN}\CC$. In $\CC^{\NN}$ bounded sets are the one included in a product of disks, whereas in $\CC^{(\NN)}$ bounded sets are included in a finite product of disks. 

\paragraph*{Duality.}
The \bornof isomorphism $(\oplus_{i\in I}E_i)^\times =\prod_{i\in I}\Ec_i$ always holds. Indeed, the restriction to each $E_i$ of a morphism $f \in (\oplus_{i\in I}E_i)^\times$ gives a family $(f_i) \in  \prod_{i\in I}\Ec_i$. Conversely, any family $(f_i) \in  \prod_{i\in I}\Ec_i$ transforms into a sum $\sum_i f_i \in (\oplus_{i\in I}E_i)^\times$ which is pointwise convergent as it is applied to finite sequences of terms.  The dual isomorphism $(\prod_{i\in I}E_i)^\times =\oplus_{i\in I}\Ec_i$ holds only in certain cases, and in particular when $I$ is countable.

\begin{prop}
\label{prop:infiniteprod}
We have $(\prod_{i \in \NN} E_i)^{\times} = \oplus_{i \in \NN} \Ec_i$.
\end{prop}

\begin{proof}
Let us first consider $ h \in  \oplus_{i \in \NN} \Ec_i$, we can define $h_i\in\Ec_i$ the $i$th components of $h$, so that $h=\sum_{i \in \NN}h_i$. As a finite sum,  $h \in (\prod_{i \in \NN} E_i)^{\times}$.

Now, consider $f \in (\prod_{i \in \NN} E_i)^{\times}$ and let us write $f_i: E_i \rightarrow \mathbb{C}$ for $f_{\mid\{0\} \times \dots \{0\} \times E_i\times \{0\} \times\dots}$, that is the restriction of $f$ to $E_i$. $f_i$ is \bornof. Let us show that there is only a finite number of $i$ such that $f_i$ is not null. Indeed, if this is not the case, there is a non decreasing sequence $(i_k) \in \NN^{\mathbb{N}}$ and for any $k \in \mathbb{N}$, $x_k \in E_{i_k}$ such that $f(0,\dots,0,x_k,0,\dots)=f_{i_k} (x_k) > k$. Remark that the set  $\{(0,\dots,0, x_k,0,\dots)\ |\ k \in \mathbb{N} \}$ is bounded in $\prod_{i \in \NN} E_i$, since $f$ is \bornof, we get a contradiction.

Let $h= \sum_{i \in \NN}  f_i$. We have just proved that $h \in \oplus_{i \in \NN} \Ec_i$, so that $h$ is \bornof as a finite sum of \bornof functions. Notice that $h=\sum_{i \in \NN} f_i \in (\prod_{i \in \NN} E_i)^{\times}$. Let us now show that $g = f-h$ is null. Remark that for any $ i\in \NN$,  the restriction of $g$ to $E_i$ is null. Suppose that $g \neq 0$. There is $x \in \prod_{i \in \NN} E_i$ such that $g(x)=0$. Consider $i$ maximal such that there is $x \in \{0\} \times \dots \{0\} \times \prod_{k \geq i} E_k$ such that $g(x) = 0$. Then $g(x) = g_i (x_i) + g_{|   \prod_{k > i} E_k} ((x_k)_{k > i})$. As $g_i (x_i)= 0$ we have $g_{|   \prod_{k > i} E_k} ((x_k)_{k > i}) =0$, and thus $g(0, \dots, 0, x_{i+1}, x_{i+2}, \dots ) = 0$. This contradicts the maximality of $i$. 
Then $g=0$, and $f = h \in  \oplus_{i \in \NN} \Ec_i$. 
\end{proof}

There is a generalization of this proposition. Thanks to the Mackey-Ulam theorem~\cite{Ula}, when the cardinal $I$ indexing the family is not strongly inaccessible, then the \bornof dual of the product is the coproduct of the \bornof duals.

\subsection{The monoidal structure}
\label{lin_mono}

\begin{defi}
The \emph{\bornof tensor product}~\cite[I.5.7]{KriMi} $E\otimes_\beta F$ is the algebraic tensor product with the finest locally convex topology such that $E\times F\rightarrow E\otimes F$ is bounded. The \emph{complete} \bornof tensor product $E\mctens F$ is the Mackey-completion of $E\otimes_\beta F$. The tensor product is associative.
\end{defi}

The bounded sets associated with this topology are generated by $b_E\otimes b_F$ for $b_E$ and $b_F$ respectively bounded in $E$ and $F$.  The unit $\unit$ is the base field $\CC$ endowed with its usual topology.

\begin{defi}
The space of linear bounded functions  $\linb(E,F)$ is  endowed with the \emph{bounded open} topology, generated by $\W(b,V)=\{f\in\linb(E,F)\,|\, f(b)\subset V\}$ where $b$ is bounded in $E$ and $V$ is open in $F$.
\end{defi}

The associated bornology is generated by the \emph{equibounded} sets, that is the $B\subset\linb(E,F)$ such that for any bounded $b$ in $E$, $B(b)$ is bounded in $F$. Indeed, consider $B\subset\linb(E,F)$ bounded for the topology of uniform convergence on bounded set. Consider $b \subset E$ a bounded set and $V \subset F$ a $0$-neighbourhood in $F$. As $B$ is bounded, there is $\lambda \in \CC$ such that $B \subset \lambda \W(b,V)$, that is $B(b) \subset \lambda V$. Thus $B(b)$ is bounded in $F$. Conversely, consider $B\subset\linb(E,F)$ an equibounded set, $b$ a bounded in $E$ and $V \subset F$ a $0$-neighbourhood in $F$. Then there is $\lambda \in \CC$ such that $B(b) \subset \lambda V$, that is  $B \subset \lambda \W(b,V)$. Thus $B$ is bounded in $\linb(E,F)$.

\begin{prop}
\label{lin_mco}
Let $E$ and $F$ be \lct. If $F$ is \mco, then so is $\linb(E,F)$.
\end{prop}
\begin{proof}
  Thanks to Proposition~\ref{prop:B(E,F)mco}, a Mackey-cauchy net in $\linb(E,F)$ converges into a bounded mapping from $E$ to $F$. Moreover, the limit of a net of linear functions is also linear.
\end{proof}

Let $E,F,G$ be locally convex spaces. Endowed with the bounded open topology, the space of \emph{bounded bilinear mappings}, denoted as $\linb(E,F;G)$, is locally convex. 

\begin{theorem}
The bornological tensor product  is the solution of the universal problem of linearizing bounded bilinear mappings. More precisely, for any $h\in\linb(E,F;G)$, there is a unique $h_\beta\in\linb(E\otimes_\beta F,G)$ such that
\begin{equation*}
  \xymatrix{
    E\times F \ar[d]^h \ar[r]& E\otimes_\beta F \ar@{-->}[dl]^{h_\beta}\\
    G&
  }
\end{equation*}
\end{theorem}

\begin{proof}
Consider $E,F,G$ and $h$ as in the proposition. Let us define $h_{\beta}: x \otimes y \mapsto h(x,y)$. we see that $h_{\beta}$ is linear and \bornof. The uniqueness of $h_\beta$ follows from the universal property of $E \otimes F$ in the category of vector spaces and linear map.  
\end{proof}

If moreover  $G$ is \mco, then so is $\linb(E,F;G)$ (for the same reason as in the proof of Proposition~\ref{lin_mco}).  Then the universal property diagram can be extended through the Mackey-completion universal property, for any $h\in\linb(E,F;G)$, there is a unique $\hat h\in\linb(E\mctens F,G)$ such that
\begin{equation*}
  \xymatrix{
    E\times F \ar[d]^h \ar[r]& E\otimes_\beta F \ar@{-->}[dl]_{h_\beta}\ar[r]& E\mctens F \ar@{-->}[dll]^{\hat h}\\
    G&
  }
\end{equation*}

\begin{prop}
\label{prop:iso_tens_bilin}
 $E\otimes_\beta -$ is left adjoint to $\linb(E,-)$, \ie for any locally convex spaces $E,F$ and $G$, there are natural isomorphism
\begin{equation*}
  \Lin(E,\linb(F,G)) \simeq \linb(E,F;G) \simeq \Lin(E\otimes_\beta F,G) .
\end{equation*}
 This property extends to the complete case by the universal property of the Mackey-completion. If $E$, $F$ and $G$ are \mco, then 
\begin{equation*}
  \Lin(E,\linb(F,G)) \simeq \linb(E,F;G) \simeq \Lin(E\mctens F,G) .
\end{equation*}
\end{prop}

\begin{proof}
(See Kriegl and Michor~\cite[I.5.7]{KriMi}) The bijection $\Lin(E\otimes_\beta F,G) \simeq \linb(E,F;G)$ follows from the preceding Theorem. As $\linb(E,F;G)$ bears the topology of uniform convergence on products of bounded sets, the bijection and its inverse are bounded isomorphisms. The morphism $ f \mapsto ( x,y \mapsto f(x)(y))$ is a bijection from  $\Lin(E,\linb(F,G))$ to $\linb(E,F;G)$ which inverse if $g \mapsto (x \mapsto (y \mapsto g(x,y)))$. Both are bounded. These bijections are natural in every elements $E$, $F$ and $G$.
\end{proof}

The next theorem follows from the symmetry and the associativity of the tensor product, and from Propositions \ref{lin_mco} and \ref{prop:iso_tens_bilin}:
\begin{theorem}
The category $\Lin$ of \mco spaces endowed with the Mackey-completed tensor product $\mctens$ is symmetric monoidal closed.  
\end{theorem}


\section{$\Smo$ maps in topological vector spaces}
\label{sec:smooth}

\mco spaces are already at the heart of a model of the differential extension of the Intuitionist Linear Logic~\cite{BET10}, inspired by the work of Frölicher, Kriegl and Michor~\cite{FroKri,KriMi}. In this model, spaces are interpreted as \mco \emph{bornological} spaces, \ie spaces such that topologies and bornologies are mutually induced. Non-linear proofs are interpreted by smooth maps. 

Actually, the bornological of \cite{FroKri, BET10} condition can be released as in~\cite{KriMi}. In particular, the characterization of open sets as bornivorous sets is never used. So that, we can get rid of this condition as in the overview below. Nevertheless, constructions such as tensor product or exponential use Mackey-completion and hence give rise to bornological spaces.  

In this paragraph,  we use $\CC$ as base field instead of $\RR$ which was used in~\cite{BET10} and in the first chapter of~\cite{KriMi}.
However, as underlined in~\cite[II.7.1]{KriMi}, any complex locally convex space can be seen as a real convex space endowed with a linear complex structure $J:E\rightarrow E$ defined by $J(x)=i\,x$ and the complex scalar multiplication is then given by $(\lambda+i\,\mu)\,x =\lambda\,x+\mu\,J(x)$. The only adaptation consists in replacing absolutely convex sets by $\CC$-absolutely convex ones. Moreover, a $\CC$-linear functional $l$ is characterized by its real part $\mathrm{Re}\circ l$, since $l(x)=(\mathrm{Re}\circ l)(x)+i (\mathrm{Re}\circ l)(J(x))$. Thus, considerations on smooth curves as well as concepts used in~\cite{KriMi,BET10} still hold in the complex setting.

\subsection{Smooth curves and smooth maps}
\label{subsec:smooth_maps}

Let $E$ be a \mco space. As in any a topological space and for any curve   $c:\RR\rightarrow E$, the derivative can be defined as usual:
\begin{equation*}
  c'(t)=\lim_{s\to 0} \frac{c(t+s)-c(t)}{s}.
\end{equation*}   
Then such a  curve is \emph{smooth} whenever it is infinitely derivable. Let us write $\cC_E$ for the set of smooth curves into $E$.
It is endowed with the topology of uniform convergence on bounded sets of each derivative separately. 

\begin{prop}
$E$ is \mco if and only if so is $\cC_E$.
\end{prop}

\begin{proof}
(See~\cite[I.3.7]{KriMi}) Suppose that $E$ is complete. By considering the set of all its derivative, one can map a smooth curve on $E$ to an element of $\prod_n \mathcal{B} ( \RR, E)$. This mapping is bounded. The image of $\cC_E$ in $\prod_n \mathcal{B} ( \RR, E)$ is closed when the last one is endowed by the product topology (see Lemma \cite[I.3.5]{KriMi}), and $\prod_n \mathcal{B} ( \RR, E)$ is complete as $E$ is (see Propositions \ref{prop:stability_mco} and \ref{prop:B(E,F)mco}). Thus $\cC_E$ is \mco. 

Suppose now that $\cC_E$ is \mco. E can be identified with the closed subspace of $\cC_E$ given by the constant curves, and thus is \mco.
\end{proof}

A set of curves $C\subset\cC_E$ is bounded whenever each derivative is uniformly bounded on bounded subsets of $\RR$ (see~\cite[I.3.9]{KriMi}):
$$\forall i, \forall  b\subset\RR \text{ bounded},\ \exists b_E\text{ bounded in }E, \text{ such that } \{c^{(i)}(x)\ |\ c\in C,\ x\in b\}\subset b_E.$$

Let $\Cin(E,F)$ denote the space of smooth maps between $E$ and $F$, \ie functions $f:E\to F$ which preserve smooth curves: $\forall c\in\cC_E,\ f\circ c\in\cC_F$. This definition of smoothness is a generalization (see~\cite{Bom}) of the usual definition for finite dimension topological vector spaces. 

\begin{prop}
When $F$ is \mco, then $\Cin(E,F)$ is also \mco.
\end{prop} 

\begin{proof}
We explain a proof of Kriegl and Michor~\cite[I.3.11]{KriMi}. By definition, $\Cin(E,F)$ bears the topology induced by the product topology.$ \prod_{c\in\cC_E} \cC_F$.  Moreover, smooth maps corresponds exactly to elements $(f_c)_c$ of $\prod_{c\in\cC_E} \cC_F$ such that for every $g \in \mathcal{C}^{\infty}(\RR, \RR)$, $f_{c \circ g} = f_c \circ g$. This set is closed for the product topology, so $\Cin(E,F)$ is \mco.
\end{proof}

A subset $B$ of $\Cin(E,F)$ is bounded whenever,   for any curve $c\in\cC_E$, its image $c^\ast(B)=\{f\circ c\ |\ c\in B\}$ is bounded in $\cC_F$.

There is a strong link between boundedness and smoothness. 
First, smoothness only depend on the bounded subsets (see~\cite[I.1.8]{KriMi}). So that, if two different topologies on $E$ induce the same bounded subsets, then the set of smooth curves into $E$ are identical.
Moreover, the space of bounded linear maps can be embedded in the space of smooth ones:
\begin{prop}{\cite[I.2.11]{KriMi}}
\label{bounded_lin_smooth}
The linear \bornof maps between $E$ and $F$ are exactly the smooth linear ones. 
\end{prop}
This is not the case for continuity and boundedness. Indeed, a \bornof linear map has not to be continuous. Consider for example an infinite dimensional Banach space $B$, and the same space but endowed with its weak topology $B_w$. By Lemma \ref{scalbound}, the identity function $\mathrm{id} : B_w \rightarrow B$ is bounded. But as the weak topology is strictly coarser than the norm topology, $\mathrm{id}$ is not continuous. Though, any continuous linear map is bounded and so smooth. 

Notice that the bounded open topology of $\linb(E,F)$ coincide with the topology induced by $\Cin(E,F)$ (see~\cite[I.5.3]{KriMi}), so that $\linb(E,F)$ can be seen as a closed linear  subspace of $\Cin(E,F)$ (see~\cite[I.3.17]{KriMi}). 

\subsection{A model of Differential Linear Logic}
\label{subsec:DiLL}
 
One of the great interest of smooth maps as defined above is that they lead to a cartesian closed category~\cite[I.3.12]{KriMi}. Let $\Smo$ denote the cartesian closed category  of \emph{\mco spaces} and \emph{smooth maps}.  The authors of~\cite{BET10} show that between $\Lin$ and $\Smo$ can be defined a linear non-linear adjunction, thus defining a model of Intuitionistic Linear Logic. The exponential of this model carries a structure rich enough to interpret Intuitionistic Differential Linear Logic (Intuitionistic DiLL), thus giving a smooth interpretation to the syntactic differentiation of DiLL.

\paragraph{An adjunction between $\Lin$ and $\Smo$.}

Models of linear logic stems from a linear non-linear adjunction (see~\cite{Mell08} for an overview). This adjunction relates a category of spaces and linear maps, and a category of spaces and non-linear maps. One way of constructing such a  model of Linear Logic, is first to consider a monoidal closed category of linear proofs, while the other is the cartesian closed category of non-linear proofs. So as to get closer to the intuitions of DiLL,~\cite{BET10} construct a non-linear category of smooth maps, using the specific terminology of Fr\"olicher, Kriegl and Michor~\cite{FroKri, KriMi}. The adjunction stems from 
an exponential modality constructed thanks to basic tools of Distribution theory. Let us describe these constructions as they were introduced in~\cite{BET10}.


Let us introduce the \emph{Dirac delta distribution} $\delta$. For any \mco space $E$ and $x\in E$, we define 
\begin{equation*}
\delta :
\left\lbrace
\begin{split}
E &  \rightarrow \Cin(E, \CC)^\times \\
x & \mapsto \delta_x : f \mapsto f(x) 
\end{split}
\right.
\end{equation*}
$\delta$ is bounded and linear, and well defined~\cite[Lem.5.1]{BET10}.

The use of the Dirac delta function to construct the exponential can be explained. The goal in~\cite{BET10} is to construct a model of Intuitionist Linear Logic with a smooth interpretation of the non-linear proofs.  Smoothness must then be captured in the exponential $!E$. $\delta$ is typically the chameleon function in analysis, as it is smooth when applied to smooth function, bounded when applied to bounded functions, or analytic when applied to analytic functions. The construction of~\cite{BET10} is based on the fact that $\delta$ is smooth. In Section~\ref{sec:holo} we construct a function $\delta$ from $E$ to the dual of a space of power series, $\delta$ will be a power series too. 


In \cite{BET10}, the Dirac delta distributions are linearly independent (see~\cite[Lem.5.3]{BET10}). Hence, they form a basis of the linear span of the set $\delta(E)=\{\delta_x\ |\ x\in E\}$. The authors then consider the Mackey-closure of this linear subspace of $\Cin(E,\CC)^\times$ and get a \mco space that they denote $!E$. We will apply the same methods here.

Let $f\in\Lin(E,F)$ be a smooth map. Its exponential $!f\in\Lin(!E,!F)$ is defined on the set $\delta(E)$ by $$!f(\delta_x)=\delta_{f(x)}.$$ It is then extended to the linear span of $\delta(E)$ by linearity and to $!E$ by the universal property of the Mackey-completion.

The exponential functor $!$ enjoys a structure of comonad, which is defined on the Dirac delta distributions and then extended:
the counit $\epsilon$ is the natural transformation given by the linear map
$\epsilon_E\in\Lin(!E,E)$, defined by $\epsilon(\delta_x)=x$, 
the comultiplication $\rho$ has components $\rho_E\in\Lin(!E,!!E)$ given by
$\rho_E(\delta_x)=\delta_{\delta_x}$.

\begin{theorem}{\cite{BET10}}
  The cokleisli category of the comonad $!$ over  $\Lin$ is the category $\Smo$. In particular, for any \mco spaces $E$ and $F$,
  $$\Lin(!E,F)\simeq \Smo(E,F).$$
\end{theorem}


\paragraph{A differential category.}
Working with smooth functions allows the author of~\cite{BET10} to introduce a notion of differentiation, which coincides with the usual notion. This makes $\Lin$, endowed with $!$, a differential category~\cite{BCS06}, and a model of the Intuitionistic part of Differential Linear Logic~\cite{Ehr11} (DiLL). 
Indeed, Differential Linear Logic differs from linear logic by a more symmetric exponential group, where the usual promotion rule is replaced by three new rules: co-weakening, co-dereliction, and co-contraction (see Figure~\ref{fig:exp_group}). Differential categories, and their co-kleisli counterpart,  the cartesian differential categories~\cite{BCS09}, are thought of as axiomatizing the structure necessary to perform differential calculus. Models of Differential Linear Logic are basically differential categories which are also models of differential calculus, and whose exponential is endowed with a bialgebraic structure.

\begin{figure}[!h]

\begin{itemize}

\item

The Exponential Group of Linear Logic :

\begin{minipage}[c]{0.45\linewidth}
\AxiomC{$\vdash\Gamma,?A,?A$}
\RightLabel{\small (contraction)}
\UnaryInfC{$\vdash\Gamma,?A$}
\DisplayProof
\end{minipage} \hfill
\begin{minipage}[c]{0.45\linewidth}
\AxiomC{$\vdash\Gamma$}
\RightLabel{\small (weakening)}
\UnaryInfC{$\vdash\Gamma,?A$}
\DisplayProof
\end{minipage}

\vspace{0.5cm}

\begin{minipage}[c]{0.45\linewidth}
\AxiomC{$\vdash\Gamma,A$}
\RightLabel{\small (dereliction)}
\UnaryInfC{$\vdash\Gamma,?A$}
\DisplayProof
\end{minipage} \hfill
\begin{minipage}[c]{0.45\linewidth}
\AxiomC{$\vdash ? \Gamma, A$}
\RightLabel{\small (promotion)}
\UnaryInfC{$\vdash ? \Gamma, !A$}
\DisplayProof
\end{minipage}

\vspace{1cm}

\item The Exponential Group of Differential Linear Logic:

\begin{minipage}[c]{0.45\linewidth}
\AxiomC{$\vdash\Gamma,?A,?A$}
\RightLabel{\small (contraction)}
\UnaryInfC{$\vdash\Gamma,?A$}
\DisplayProof
\end{minipage} \hfill
\begin{minipage}[c]{0.45\linewidth}
\AxiomC{$\vdash\Gamma,\oc A,\oc A$}
\RightLabel{\small (co-contraction)}
\UnaryInfC{$\vdash\Gamma, \oc A$}
\DisplayProof
\end{minipage}

\vspace{0.5cm}

\begin{minipage}[c]{0.45\linewidth}
\AxiomC{$\vdash\Gamma$}
\RightLabel{\small (weakening)}
\UnaryInfC{$\vdash\Gamma,?A$}
\DisplayProof
\end{minipage} \hfill
\begin{minipage}[c]{0.45\linewidth}
\AxiomC{$\vdash\Gamma$}
\RightLabel{\small (co-weakening)}
\UnaryInfC{$\vdash\Gamma,\oc A$}
\DisplayProof
\end{minipage}

\vspace{0.5cm}

\begin{minipage}[c]{0.45\linewidth}
\AxiomC{$\vdash\Gamma,A$}
\RightLabel{\small (dereliction)}
\UnaryInfC{$\vdash\Gamma,?A$}
\DisplayProof
\end{minipage} \hfill
\begin{minipage}[c]{0.45\linewidth}
\AxiomC{$\vdash\Gamma,A$}
\RightLabel{\small (co-dereliction)}
\UnaryInfC{$\vdash\Gamma,\oc A$}
\DisplayProof
\end{minipage}

\end{itemize}

\caption{Exponential groups of LL and DiLL}
\label{fig:exp_group}
\end{figure}

Let us present the structure of bialgebra of $!$ in $\Smo$, and how the differentiation is interpreted in this category. In $\Smo$, finite products coincide with finite coproducts. This biproduct structure is transported by the strong monoidal functor $!$ to a bialgebra structure:  $\Delta:    !   E\rightarrow   !    E\mctens!   E$   is defined on Dirac distributions by   $\Delta(\delta_x)=\delta_x \otimes \delta_x$,  $e: ! E\rightarrow \CC$ is defined as $e(\delta_x)=1$,   $\nabla:    !   E\mctens!   E\rightarrow    !   E$   is given by $\nabla(\delta_x \otimes \delta_y)=\delta_{x+y}$  and $m^0: \CC\rightarrow! E$ is defined as $m^0(1)=\delta_0$. 
  
Differentiation can be constructed from the bialgebra structure and from a more primitive differentiation operator, denoted as $\mathrm{coder} \in\Lin(E, !E)$. This operator is  the interpretation of the codereliction rule of DiLL. It corresponds to the differentiation at $0$ of a smooth map :

\begin{equation*}
\mathrm{coder}(v)=\lim_{t\rightarrow 0}\frac{\delta_{tv}-\delta_0}{t}
\end{equation*}

The differential operator is then interpreted as the usual one in analysis:

  \[d: \Cin(E,F)\rightarrow \Cin(E,\Lin(E,F))\]
  \[df(x)(v)=\lim_{t\rightarrow 0} \frac{f(x+tv)-f(x)}{t}\]

\section{A quantitative model of Linear Logic}
\label{sec:pws}

The purpose of this paper is to define a new quantitative model of DiLL, with a strong analytical flavour.  Indeed, one of the characteristic of the quantitative models~\cite{Gir88,2aaf,danosehrhard,Ehr02,Ehr05} is that the morphisms in the cokleisli enjoy a Taylor expansion. The authors of~\cite{BET10} constructed a smooth interpretation of DiLL, that we would like to refine into a quantitative model.  We could have used a study of holomorphic and real analytic maps by Kriegl and Michor~\cite[Chapter II]{KriMi}: the construction of a model of holomorphic or real analytic maps should be easily done by following the constructions of~\cite{BET10}. However, these maps corresponds only locally to their Taylor development. As the interpretation of locality in denotational semantics remains unclear, we want to interpret the non-linear proofs of DiLL as functions corresponding in every point with their Taylor development in $0$. 

We take advantage of the fact that our spaces are \mco so as to define a very general notion of power series which   are in particular smooth (see Proposition~\ref{bounded_incl_pws_smo}). A power series is a converging sum of monomials. 
Indeed a power series in $\CC$ is represented by a sum 
$ \sum_n a_n x^n $
converging pointwise on some disk. We are going to use power series between topological vector spaces, thus the description has to be a little bit more involved and a power series will be a sum
$ \sum_n f_n $
where $f_n$ is $n$-homogeneous and $\sum_n f_n (x)$ converges for every $x \in E$. Moreover, we need a stronger notion than pointwise convergence, so as to compose power series and to get a cartesian closed category. This is the uniform convergence on bounded sets of the partial sums $\sum_{n=0}^N f_n$, which will allow us to deeply relate weak, strong and pointwise convergence of power series (see Proposition~\ref{wcv_pointwise_pws}).
As the space of power series between \mco spaces is \mco (see Proposition \ref{prop:compl_series}), we obtain a cartesian closed category of \mco spaces and power series between them. 

To get to this point, we use a description of power series as functions sending holomorphic maps on holomorphic maps, and for this, proofs of~\cite{KriMi} are adapted. This study gives us a Cauchy inequality on power series, and equivalences between weak convergence and strong convergence of power series, inspired from~\cite{BS71}. Finally, using weak convergence, we obtain the cartesian closedeness of the category.

The part on holomorphic maps between \lct is not needed at first reading, as it only results into Proposition~\ref{Cauchy_ineq_pws}. The reader may then skip Section~\ref{sec:holo}.

Due to the connection between power series and holomorphic maps, we consider vector spaces over $\CC$. In the following, $\DD$ denotes the closed unit ball in $\CC$ and $E$, $F$, and $G$ range over \mco spaces.

\subsection{Monomials and power series}

\begin{defi}
A function $f_n: E \rightarrow F$ is an \emph{$n$-monomial} when there is $\tilde{f_n}$ an $n$-linear \bornof function from $E$ to $F$ such that for each $ x \in E$ $$f_n(x) = \tilde{f_n} \underbrace{(x,...,x)}_{n \text{ times}}.$$ 

We write $\nlin(E,F)$ for the space of $k$-monomials from $E$ to $F$, and $\linb(E^{\otimes^n}, F)$ for the space of \bornof $n$-linear maps from $E$ to $F$. We endow $\linb(E^{\otimes^n}, F)$ with the locally convex topology of uniform convergence on bounded sets of $E$. As in the linear case (see Section~\ref{lin_mono}), bounded sets of $\nlin(E,F)$ are the equibounded ones.
\end{defi}

The following   formula relates the values of a monomial with the values of the unique multilinear map it comes from.
\begin{lemma}
\label{polarization_formula}
Consider $f_n \in \linb^n(E,F)$, and consider $\tilde{f_n}$ an $n$-linear map such that $\tilde{f_n}(x, \dots, x) = f_n (x)$. Then for every $x_1,\dots, x_n\in  E$:
$$\tilde f_n(x_1,\dots,x_n)= \tfrac 1{n!} \sum_{\epsilon_1,\dots,\epsilon_n =0}^1 (-1)^{n-\sum_{j=1}^n \epsilon_j} f_n\left(\sum_{j=1}^n\epsilon_jx_j\right).$$
\end{lemma}
\begin{proof}
 The proof relies on the expansion of the right hand side by multilinearity and symmetry of $\tilde f_n$ (see~\cite[II.7.13]{KriMi}).
\end{proof}

As in the case of bounded linear functions (see Proposition~\ref{bounded_Mackey}) monomials behave particularly well with respect to Mackey-convergence.
\begin{lemma}
\label{lem:lbdd-MC} Consider $(x_{\gamma})_{\gamma \in \Gamma}$ a Mackey-converging net in $E$ and $f_k : E \rightarrow F$ a $k$-monomial. Then $f_k(x_{\gamma})$ is a Mackey-converging net, thus a converging net. 
\end{lemma}
\begin{proof}
Let us write $\tilde{f}_k$ for the symmetric bounded $k$-linear map corresponding to $f_k$. Let $b \subset E$ be a bounded set, $x \in E$ and  $(\lambda_{\gamma \in \Gamma}) \in \CC ^{\mathbb{N}}$ be a sequence decreasing towards $0$ such that for every $\gamma$ :

$$x_{\gamma} -x \in \lambda_{\gamma} b. $$

Let us write $b'= \tilde{f_{k}} (b \times ... \times b)$. Then for every $\gamma\in \Gamma$, we have can factorize $f_k(x_{\gamma}) -f_k(x)$ following the classical equality $x^k - y^k = (x-y)(x^{k-1} + x^{k-2}y + ... + y^{k-1})$. Indeed:
$$ f_k(x_{\gamma}) -f_k(x) = \tilde{f_k}(x_{\gamma} -x ,x_{\gamma},..., x_{\gamma})   + \tilde{f_k}(x_{\gamma} -x ,..,x_{\gamma},..., x_{\gamma},x) + \tilde{f_k}(x_{\gamma}-x, x, ..., x) .$$

As $\tilde{f_k}$ is bounded, and as for every $\gamma \in \Gamma$ $x_\gamma$ belongs to the bounded set $M b+\{x\}$ for some M, there is  a bounded $b'$ in $F$ such that, for every $\gamma$:
$$ f_k(x_{\gamma}) -f_k(x) \in \lambda_{\gamma} b'.$$

\end{proof}

An \emph{$n$-homogeneous} function is a map $f$ such that $f ( \lambda x) = \lambda^n f(x)$ for any scalar $\lambda$. 
\begin{lemma}{\cite[I.5.16.1]{KriMi}}
\label{lem:monom_smooth}
A function $f$ from $E$ to $F$ is an $n$-monomial if and only if it is a smooth $n$-homogeneous map. 
\end{lemma}
\begin{proof}
As bounded $n$-linear functions  are smooth by Proposition~\ref{bounded_lin_smooth},  $n$-monomials are smooth $n$-homogeneous functions. The converse is also true. Indeed, by deriving at $0$ an $n$-homogeneous smooth function along the curve $t\mapsto tx$, we can show that it is equal to its $n^{th}$-derivative which is $n$-linear.
\end{proof}

\begin{prop}
\label{nlinmco}
If $F$ is \mco, then so is $\nlin(E,F)$. 
\end{prop}
\begin{proof}
There is a bounded isomorphism between the space $\nlin(E,F)$ and the space of all $n$-linear symmetrical morphisms from $E$ to $F$, when the last one is endowed with the topology of uniform convergence on bounded sets of $E \times ... \times E$. Indeed, one associate an $n$-monomial to an $n$-linear symmectric morphism by applying the last one $n$-times to the same argument. Thanks to the Polarization Formula \ref{polarization_formula}, we can  obtain an $n$-linear symmetric morphism $\tilde{f_n}$ from an $n$-monomial $f_n$:
$$\tilde f_n(x_1,\dots,x_n)= \tfrac 1{n!} \sum_{\epsilon_1,\dots,\epsilon_n =0}^1 (-1)^{n-\sum_{j=1}^n \epsilon_j} f_n\left(\sum_{j=1}^n\epsilon_jx_j\right).$$

The mappings $(f_n \mapsto \tilde{f_n})$ and $(\tilde{f_n} \mapsto f_n)$ preserves uniformly bounded sets, thus $\nlin(E,F)$ and the space of all $n$-linear symmetrical morphisms from $E$ to $F$ are isomorphic. 
By definition of the symmetrized $n$-th product $E^{\otimes^n_s}$, the space $\nlin(E,F)$ is also isomorphic to $\linb (E^{\otimes^n_s},F)$. This space is \mco as $F$ is (see Proposition~\ref{lin_mco}), and thus  $\nlin(E,F)$ is also \mco.

\end{proof}

\begin{defi}
A polynomial function is a finite sum of monomials : $$\forall x,\ P(x) = \sum_{n=0}^N f_n(x).$$

We write $\pol (E,F)$ for the space of all polynomial functions between $E$ and $F$, and endow it with the topology of uniform convergence on bounded subsets of $E$.
\end{defi}


\begin{defi}\label{def:powerseries}
A function $f$ from $E$ to $F$ is a \emph{power series} when $f$ is \emph{pointwise} equal to a converging sum of $k$-monomials:
$$\forall x, \ f(x)= \sum\limits_{k=0}^{\infty} f_k(x), $$
and when this sum converges \emph{uniformly on bounded} sets of $E$.

We write $\series(E,F)$ for the space of power series between $E$ and $F$ and endow it with the topology of uniform convergence on bounded subsets of $E$.
\end{defi}

\begin{prop}
\label{pws_borno}
A power series is \bornof.
\end{prop}
\begin{proof}
Consider $ f = \sum_k f_k \in \series(E,F)$, $b$ a bounded set in $E$, and $U$ an absolutely convex $0$-neighbourhood in $F$.
We know that $\sum_k f_k$ converges uniformly on b. Hence, there is an integer $N$ such that $( f - \sum_{k=1}^N f_k ) (b) \subset U$. 
Besides, each $f_k$ sends $b$ on a bounded sets, thus $(\sum_{k=1}^N f_k ) (b)$ is bounded as a finite sum of bounded sets. So there is $\lambda\in\CC$ such that $(\sum_{k=1}^N f_k ) (b)\subset \lambda U$.
Finally,  $ f (b) \subset (\lambda+1) U$.
\end{proof}

\subsection{Power series and holomorphy}
\label{sec:holo}

We need to study the power series we defined more deeply. We are going to show that if $f = \sum_n f_n : E \rightarrow F$ is a power series converging uniformly on bounded sets, it is holomorphic, according to the specific definition of Kriegl and Michor~\cite[II.7.19]{KriMi}. This definition is a generalisation of the well known definition of holomorphy for complex functions of a complex variable, and leads to a Cauchy inequality for $f$ (see Proposition~\ref{Cauchy_ineq_pws}). This Cauchy inequality will turn to be essential in showing cartesian closedeness and the composition results in Section~\ref{subsec:cv_pws}.

This formula will in particular result in the Mackey-convergence of power series (see Proposition~\ref{MC_pws}), and will allow us to compose \bornof linear forms with power series (see Proposition~\ref{comp_lin_pws}). For now on, we are going to work with linear continuous forms $l\in E'$ in order to be able write $ l \circ ( \sum_n f_n) = \sum_n l \circ f_n$. 

\paragraph*{Holomorphic curves in $\mathbb{C}$.}
Remember that an holomorphic curve $c: \mathbb{C} \rightarrow \mathbb{C} $ is a complex everywhere derivable function. It is then infinitely many times differentiable, and verifies the Cauchy formula and the Cauchy inequality. For any $z'\in\CC$ and any sufficiently small $r$:
\begin{gather*}
\frac{c^{(n)}(z')}{n!} = \frac{1}{2 \pi i } \int\limits_{|z-z'|=r} \frac{c(z)}{(z-z')^{n+1}} d z\label{eq:CF} \\
\intertext{and hence :}
\left| \frac{c^{(n)}(z')}{n !} \right| \leq \left| \frac{sup \{  c(z)\ |\ |z-z'| = r \} }{ r ^n } \right|. \label{eq:CI}
\end{gather*}

Moreover, it can be uniquely decomposed as a power series:
\begin{equation*}
  \forall a \in \CC, \forall z\in \CC,\ c(z+a)=\sum_n \frac{c^{(n)}(a)}{n!}z^n.
\end{equation*}

\paragraph*{Holomorphic curve in an \lct.}

This part on holomorphic curve is inspired by  Part 7 of the book of Kriegl and Michor~\cite{KriMi} on \mco spaces and holomorphic functions  and by the first theorem of~\cite{Gro53}. 
We give two different approaches to holomorphic curves, that we then show equivalent.

\begin{defi}
 A \emph{strong holomorphic} curve $c: \mathbb{C} \rightarrow E$ is an everywhere complex derivable function.
A \emph{weak holomorphic} curve $c: \mathbb{C} \rightarrow E$ is a function such that for every $ l \in E'$, $l \circ c $ is holomorphic.
\end{defi}

\begin{lemma}\label{lem:properties-holomorphic}
Let $c:\CC\rightarrow E$ be a curve.
\begin{enumerate}
\item\label{property:bounded} If $c$ is strong holomorphic, then 
\begin{center}
$\forall l\in E'$, $l\circ c$ is complex derivable and $\forall z\in\CC,\ (l\circ c)'(z)= l (c'(z))$.
\end{center}
\item\label{property:whbdd} If $c$ is weak holomorphic, then $c$ is bounded.
\item\label{property:whMC} If $c$ is weak holomorphic, then for all $z\in\CC$, the difference quotient $(\frac{c(z+h)-c(z)}{h})_{h\in \mathbb D}$ is a Mackey-Cauchy net.
\end{enumerate}
\end{lemma}
\begin{proof}
  Let $c$ be a strong holomorphic curve.

  \ref{property:bounded}.
  Let $l\in E'$. Since $l$ is linear and  continuous, we have:
    \begin{equation*}
      \lim\limits_{h \rightarrow 0} \frac{l \circ c (z + h) - l \circ c (z)}{h} = l(c'(z)).
    \end{equation*}
    Then, $l\circ c$ is complex derivable and $\forall z\in\CC,\ (l\circ c)'(z)= l (c'(z))$.
  
    \medskip
  Now, let $c$ be a weak holomorphic curve.

  \smallskip
  \ref{property:whbdd}.
  Let $b$ be a bounded set in $\mathbb{C}$ and $\bar b$ its
    closed absolutely convex closure. For every $ l \in E'$, $(l\circ
    c)( \bar{b})$ is compact as the image in $ \mathbb{C}$ of a
    compact set by a continuous function ($l\circ c$ is complex
    holomorphic and thus continuous). Then, $c(b)$ is weakly bounded and
    so bounded by Proposition~\ref{scalbound}.

  \medskip
  \ref{property:whMC}
 This proof is adapted from~\cite[I.2.1]{KriMi}. By translating
    $c$, we may assume that $z=0$. For any $l\in E'$, $l\circ c$ is
    holomorphic in $\CC$, hence infinitely complex-derivable and $l\circ c$ is Lipschitz continuous. Then, we have

  \begin{eqnarray*}
    \frac{1}{z_1-z_2}\left(\frac{l \circ c(z_1)- l \circ c(0)}{z_1}- \frac{l \circ c(z_2)- l \circ c(0)}{z_2} \right)
    & =& \int_0^1 \frac{(l \circ c)'(rz_1)-(l \circ c)'(rz_2)}{z_1-z_2} dr \\
    & =& \int_0^1 \frac{(l \circ c)'(rz_1)-(l \circ c)'(rz_2)}{rz_1-rz_2} rdr 
  \end{eqnarray*}

 Moreover the curve $r\mapsto \frac{(l \circ c)'(rz_1)-(l \circ c)'(rz_2)}{rz_1-rz_2}$ is locally bounded as $(l \circ c)$ is holomorphic. The set $\left\{ \frac{1}{z_1-z_2}\left(\frac{ c(z_1)-  c(0)}{z_1}- \frac{ c(z_2)- c(0)}{z_2}) \right)\mid z_1,z_2\in \mathbb D \right\} $ is then scalarly bounded and thus bounded by Proposition~\ref{scalbound}. This is equivalent to show that the difference quotient is Mackey-Cauchy (see Definition~\ref{def:mackey}).

\end{proof}

\begin{prop}
\label{weak_cont_strong_holocurve}  
\label{strong_weak_holo}
The strong holomorphic curves into a \mco space are exactly the weak holomorphic curves.
\end{prop}
\begin{proof}
  A strong holomorphic curve is weak holomorphic by the first property of the preceding lemma.

  Now, let $c:\CC\rightarrow E$ be a weak holomorphic curve into a \mco space $E$. Then by the third property of the preceding lemma, for all $z\in\CC$, the difference quotient $(\frac{c(z+h)-c(z)}{h})_{h\in\mathbb D}$ is Mackey-Cauchy and thus converges in $E$,  since it is \mco. Hence, $c$ is complex derivable and its derivative $c'(z)$ is the limit of the difference quotient.
\end{proof}
From now on, an \emph{holomorphic} curve is either a weak or strong holomorphic curve.

\begin{lemma} 
\label{lem:int-abs-conv}
  Let $b$ be an absolutely convex and closed subset of $E$, $\gamma$ be a path in $\CC$ and $f: \CC\to E$ be continuous. If for any $z\in\gamma([0;1])$, $f(z)\in b$, then the integral of $f$ on the path $\gamma$ is in $b$.
\end{lemma}

\begin{proof}
As $\int_{\gamma} f = \int_0^1 f(\gamma(t)) dt$, this integral can be computed as the limit of the Riemann sums over $[0;1]$ of $f \circ \gamma$. As $b$ is absolutely convex, each of these sums is in $b$. As it is closed, we have also $\int_{\gamma} f  \in b $.
\end{proof}

\begin{prop}
\label{Cauchy_ineq_curves}
Let $c:\CC\rightarrow E$ be a holomorphic curve. There is $b$ absolutely convex, closed bounded subset of $E$ such that:
\begin{center}
 $c ( \mathbb{D} ) \subset b $ and $\forall n\in\mathbb N$, $ c^{(n)}(\mathbb{D}) \subset n ! b $.
\end{center}

\end{prop}

\begin{proof}
Thanks to Property~\ref{property:bounded} of Lemma~\ref{lem:properties-holomorphic}, $c$ is bounded. This justifies the existence of $b$ such that $c ( \mathbb{D} ) \subset b $.
Moreover for every $l \in E' $,  the curve $l \circ c $ is holomorphic in $\CC$ according to Proposition~\ref{strong_weak_holo}. Thus for every $z \in \CC$, 
$$ \frac{(l \circ c)^{(n)}(z)}{n!} =  \frac{ l (c^{(n)}(z))}{n!} = \frac{1}{2 \pi i } \int\limits_{|h|=1} \frac{l(c(h z))}{h^{n+1}} d h .$$
Thus $l \circ c^{(n)} ( \mathbb{D}) \subset l (n! b)$ (see Lemma~\ref{lem:int-abs-conv}).
 By the Hahn-Banach separation theorem applied to $b$ and to every $\{z\}$ for $z \in \mathbb{D}$, we get that $c^{(n)} ( \mathbb{D}) \subset n! b$.
\end{proof}

\begin{prop}\label{prop:holo_curve_dec}
Let $c:\mathbb C\to E$ be a holomorphic curve.  For any $z  \in \mathbb{C}$, $c^{(n)} (z) \in E$ and $c$ can be uniquely decomposed as a series uniformly converging on bounded disks of $\CC$:
$$c: z\mapsto  \sum\limits_n \frac{1}{n!}c^{(n)}(0) z^n .$$
Moreover, this series is Mackey-converging at each point of $\CC$.
\end{prop}
\begin{proof}
For every $l \in E'$, $l \circ c$ is a holomorphic function from $\CC$ to $\CC$. It does thus correspond in every point to its Taylor series in $0$, and as $c^{(n)} (z) \in E$ for every $z$ we have 

\begin{equation*}
 l \circ c (z)   = \sum\limits_n \frac{(l \circ c )^{(n)}(0)}{n!} z^n 
  = \sum\limits_n \frac{l  (c ^{(n)}(0))}{n!} z^n
  = l (\sum\limits_n \frac{ c ^{(n)}(0)}{n!} z^n)
\end{equation*}
As $E'$ is point separating, we have for every $z \in \CC $ : 
$$c(z) = \sum\limits_n \frac{1}{n!}c^{(n)}(0) z^n .$$

Moreover, for any $r > 0$, the closed and absolutely convex closure $b_r$ of the set $\{ \frac{1}{n!}c^{(n)}(0) r^n \ |\  n \in \mathbb{N} \}$ is bounded. It is indeed weakly bounded as the power series $ \sum\limits_n \frac{l  (c ^{(n)}(0))}{n!} z^n$ converges uniformely on the open disk of center $0$ and radius $r$. 
Thus for every $\abs{z} <r$, we have:
$$ \sum_{n \geq N} \frac{1}{n!}c^{(n)}(0) z^n \in \sum_{n \geq N} \left(\frac{\abs z}{r}\right) b_r \subset \left(\frac{\abs z}{r}\right)^N \frac{1}{1-\frac{\abs z}{r}} b_r$$
and the series   $\sum_n \frac{1}{n!}c^{(n)}(0) z^n$ does Mackey-converge towards $c(z)$.


\end{proof}

\paragraph*{Power series and holomorphy.}
The goal of this paragraph is to prove that power series, as presented in Definition~\ref{def:powerseries}, preserve holomorphic curves (see Theorem~\ref{powerseries_holo}). This will show that they follow the same pattern as smooth functions that preserve smooth curves. As mentioned in~\cite[II.7.19.6]{KriMi}, functions preserving holomorphic curves on $\DD$ are locally power series, but we do not know if the preservation of holomorphic curves characterizes our power series.

The following property is adapted from~\cite[II.7.6]{KriMi}. 

\begin{lemma}
\label{fact_Banach_holocrv}
A holomorphic curve into $E$ locally factors through a Banach space $E_b$ generated by a bounded set $b\subset E$  (see Definition~\ref{E_b}).
\end{lemma}

\begin{proof}
Consider $c$ a holomorphic curve, $z \in \mathbb{C}$ and $w$ a compact neighbourhood of $z$.  Let us denote $b$ the absolutely convex closed closure of $c(w)$. For any $l\in E'$, the Cauchy inequality~\eqref{eq:CI} gives us for $r$ small enough
$$ \frac{r^k}{k!}(l\circ c)^{(k)}(z) \in l(b).$$
Thus for $z'$ close enough to $z$ in $\mathbb{C}$,
$$ (l\circ c)(z')=\sum\limits_{k \geq 0} \left(\frac{z-z'}{r} \right)^k \frac{r^k}{k!}(l\circ c)^{(k)}(z)\ \in\ \sum\limits_{k \geq 0} \left(\frac{z-z'}{r} \right)^k l(b).$$
Then, as E' is point separating, we get that 
$$c(z') \in \sum\limits_{k \geq 0} \left(\frac{z-z'}{r} \right)^k b.$$
And for $z'$ close enough to $z$, $ c(z') \in E_b$.
\end{proof}

Now, we want to show that for every holomorphic curve $c$, if $f : E \rightarrow  F$ is a power series, then $f \circ c$ is again a holomorphic curve (see Theorem~\ref{powerseries_holo}). Mainly, this is shown by working on the $E_b$, so as to use Banach spaces properties. Remember that we are working in \mco spaces, and that a space is \mco if and  by only if each $E_b$ is a Banach space (see Proposition \ref{E_b_mco}). This is a generalization of a result by Kriegl and Michor \cite[II.7.17]{KriMi}.
 
\begin{lemma}\label{lem:bdd-monom}
Let $f = \sum_k f_k$ be a power series from $E$ to $F$. For any bounded set $b$ of $E$, the set $\{ f_k(x)\ |\ x\in b \}$ is bounded in $F$. 
\end{lemma}
\begin{proof}
Let us write $S_n = \sum_{k \leq n} f_k$, and fix $b$ any bounded set of $E$. Then, by definition of power series, $S_n$ converges uniformly on bounded sets, hence for every $U$ neighbourhood of $0$ in $E$, there is $p$ such that if $n , m \geq p$ we have $(S_n - S_m )(b) \subset U$. In particular, for $k \geq p+1$,  $f_k (b) \subset U$. Because the $f_0,\dots, f_p$ are \bornof, then there are $\lambda_0, \dots, \lambda_k \in \mathbb{C}$ such that $ f_j(b) \in \lambda_j U$. Finally, we get $\{f_k\, |\, k \in \mathbb{N}\} (b) \subset  \max\{1,\lambda_0,\dots, \lambda_k \} U$.

\end{proof}

\begin{theorem}
\label{powerseries_holo}
Power series send holomorphic curves on holomorphic curves.
\end{theorem}
\begin{proof}
Let $ f = \sum_k f_k: E \rightarrow F$ be a power series, and $c: \mathbb{C} \rightarrow E$ be a holomorphic curve. Let $\tilde f_k$ be the $k$-linear bounded map associated to the $k$-monomial $f_k$. 

Let us show that the curve $ f \circ c: \mathbb{C} \rightarrow F$ is holomorphic. Thanks to Proposition~\ref{weak_cont_strong_holocurve}, it is enough to show that for every $l \in F'$, $ l \circ f \circ c: \mathbb{C} \rightarrow \mathbb{C}$ is holomorphic. Let us fix $z_0 \in \mathbb{C}$ and show that locally around $z_0$, $l \circ f \circ c $ is complex derivable. By translating $c$, we can assume \wlg that $c(z_0) = 0$, and $z_0 =0$. Besides, by Proposition~\ref{fact_Banach_holocrv} we can assume \wlg that $E$ is a Banach space. 

Thanks to Propostion~\ref{prop:holo_curve_dec}, we can write locally $c$ as a Mackey-converging power series in $E$. For every $z \in \CC$ we have:
$$c(z) = \sum\limits_n a_n z^n .$$ 
Moreover, this series converges uniformely on $\DD$.

Because $l$ is linear and continuous, we have $l \circ f = \sum_k l \circ f_k $. Besides, for any $k\in \mathbb N$, $l\circ \tilde f_k$ is $k$-linear and bounded. Thanks to Lemma~\ref{lem:lbdd-MC}, $\sum\limits_{n_1} \dots \sum\limits_{n_k} l \circ \tilde f_k( a_{n_1}, \dots, a_{n_k}) z^{n_1+\dots+n_k}$ converges to $l\circ f_k(c(z))=l\circ \tilde f_k(c(z),\dots,c(z))$. We thus have:
$$l \circ f(c(z)) = \sum\limits_k \sum\limits_{n_1} \dots \sum\limits_{n_k} l \circ \tilde f_k( a_{n_1}, \dots ,a_{n_k}) z^{n_1+\dots+ n_k} $$

Let us now apply Lemma~\ref{lem:bdd-monom} to the unit disk $U$, which is bounded, in the Banach space $E$. We get that $\{l \circ \tilde f_k(x_1, \dots, x_k)\ |\ k \in \mathbb{N}, x_j \in U \}$ is bounded. Since $\sum_n a_n z^n$ converges, for any $\abs z<1$ and $n$ big enough, $a_nz^n\in U$. Thus, for  $r < 1 $, we have for all $n \geq N $  $ a_n r^n \in U$, thus the following set is bounded:
$$b=\left\{ l \circ \tilde f_k(a_{n_1}r^{n_1}, \dots , a_{n_k}r^{n_k})  | n_i \geq N \right\}.$$

Following~\cite[II.7.17]{KriMi}, consider $z$ and $r$ such that $|z|< \frac{1}{2}$ and $ 2 |z |< r < 1$, then 
\begin{multline}\label{eq:absc}
\sum\limits_k \sum\limits_{n_1} \dots \sum\limits_{n_k} l \circ \tilde f_k( a_{n_1}, \dots, a_{n_k}) z^{n_1+\dots +n_k}  \\
\begin{split}
& = \sum\limits_k \sum\limits_{n_1} \dots \sum\limits_{n_k}  l \circ \tilde f_k(a_{n_1}r^{n_1}, \dots , a_{n_k}r^{n_k}) \frac{z^{n_1+\dots +n_k}}{r^{n_1+\dots +n_k}}, \\
& = \sum\limits_n \sum\limits_k \sum\limits_{n_1+\dots+n_k = n}  l \circ \tilde f_k(a_{n_1}r^{n_1}, \dots , a_{n_k}r^{n_k}) \frac{z^{n_1+\dots +n_k}}{r^{n_1+\dots +n_k}}.
\end{split}
\end{multline}
Now, we look at the last sum and get:
$$ \sum\limits_n \sum\limits_k \sum\limits_{n_1+\dots+n_k = n}  l \circ \tilde f_k(a_{n_1}r^{n_1}, \dots , a_{n_k}r^{n_k}) \frac{z^{n_1+\dots +n_k}}{r^{n_1+\dots +n_k}} \in \sum_n (2^n - 1) \left( \frac{z}{r} \right)^n b.$$
This is an absolutely converging sum, and the permutation of the sums in~\eqref{eq:absc} is justified by Fubini's thereom. Finally,  $l\circ f\circ c$ is holomorphic in $\CC$, as the sum of an absolutely converging power series.

%

.

\end{proof}

Another proof of this theorem uses Hartog's theorem~\cite[II.7.9]{KriMi}, and the fact that a \bornof $k$-monomial sends a holomorphic curve on a holomorphic curve. 

\begin{lemma}
\label{derivatives_pws}
Let $f = \sum_k f_k$ be a power series between $E$ and $F$. Then, for every $x \in E $ and $n \in \mathbb{N}$, $c:z\mapsto f(z x)$ is a holomorphic curve into $F$ whose $n$-th derivative in $0$  is $ n! f_n (x)$.
\end{lemma}

\begin{proof}
The curve $c:z\mapsto f(z x)$ is holomorphic thanks to Theorem~\ref{powerseries_holo}. Since the scalar multiplication on $E$ is continuous, the set $\{ z x\ |\ |z|< 1\}$ is bounded. By Definition~\ref{def:powerseries} of power series, $\sum_k f_k  ( z x) = \sum_k f_k(x)\,z^k$ converges uniformly on the unit disk $\mathbb D$ of $\CC$. Thanks to the uniqueness of the decomposition (see Lemma~\ref{prop:holo_curve_dec}), its $n$-th derivative is $n!f_n(x)$.
\end{proof}

\begin{coro}
\label{unique_dvp_pws}
The $k$-monomials in the development of a  power series are unique.
\end{coro}


\begin{prop}
\label{Cauchy_ineq_pws}
Every  power series $f \in \series(E,F)$ verifies a Cauchy inequality: 
if $b$ is an absolutely convex set in $E$ and if $b'$ is an absolutely convex and closed set in $E$ such that $f(b) \subset b'$, then for all $n \in \mathbb{N}$ we have also : 
$$f_n(b) \subset b'$$
\end{prop}
\begin{proof}
For every $x \in E$,  $c: z \mapsto f(zx)$ is a holomorphic curve into $F$ whose $n$-th derivative is $n!f_n(x)$ by Lemma~\ref{derivatives_pws}. For every $l \in F'$, $l\circ c$ is holomorphic and satisfies a Cauchy Formula:
$$\tfrac{1}{n!}(l\circ c)^{(n)}(0)= l( \frac{c^{(n)}(0)}{n!})= l (f_n(x)) = \frac{1}{2i \pi  } \int\limits_{|h|=1} \frac{l(f(h x))}{h^{n+1}} d h .$$

As $b$ is absolutely convex, we conclude thanks to the Hahn-Banach separation theorem that for every $z \in b$, for every $n \in \mathbb{N}$, $f_n(z) \in b'$ (see Lemma~\ref{lem:int-abs-conv}). 
\end{proof}

\subsection{Convergence of power series}
\label{subsec:cv_pws}

Thanks to the Cauchy inequality, we will show the Mackey-convergence of the partial sums of a power series. This property is fundamental in the construction of the cartesian closed category of \mco spaces and power series. It will allow for example to ensure well-composition of power series and bounded functions. 





\begin{prop}
\label{MC_pws}
If $f = \sum_n f_n$ is a power series, then its partial sums Mackey-converge towards $f$ in $\cB(E,F)$. 
\end{prop}
\begin{proof}
Let $b$ be an absolutely convex and bounded subset of $E$ and $b'$ be the absolutely convex and closed closure of $f(b)$. By Proposition~\ref{Cauchy_ineq_pws}, for all $n\in\NN$, $f_n(b)\subset b'$. As $f_n$ is $n$-homogeneous, we also  have $f_n (\tfrac 12 b) \subset \tfrac 1{2^n}b'$.

If $B$ denotes the equibounded set $\{f\in\cB(E,F)\ \mid\ f(\tfrac{1}2 b)\subset b'\}$, then $f \in B$ as $\tfrac 12b \subset b$, $f_0 \in B$ and for every $n$, $f_n \in \tfrac 1{2^n}B$. Thus for every $N\in\NN$,
$$ f- \sum_0^{N} f_n \in \sum_{n > N} \tfrac{1}{2^n} B$$
and the partial sums do Mackey-converge towards $f$.



\end{proof}

\paragraph{Simple and weak convergence}
Our definition of power series allows us to make nice connection between their weak, strong and simple convergence. This will allow us to prove the cartesian closedeness of the category of \mco spaces and power series between them. 

\begin{prop}
\label{wcv_pointwise_pws}

Let $\{f_k\,\mid\, k\in\NN\}$ be a family of $k$-monomials from $E$ to $F$. If for every $l\in \Fc$ (resp. $l\in F'$) and $x\in E$, $\sum_k l\circ f_k(x)$ converges in $\CC$, then for any $x\in E$, $\sum_k f_k(x)$ converges in $F$.
\end{prop}

\begin{proof}
Let us fix $x\in E$. By assumption, for any $l\in \Fc$, $\sum_k l \circ f_k (2 x)$ converges in $\CC$, so $\{l\circ f_k(2x) \, \mid\, k\in \NN\}$ is bounded in $\CC$. By Proposition~\ref{scalbound} (resp. by the Mackey-Ahrens Theorem), $\{f_k(2x) \,\mid\, k\in \NN\}$ is bounded in $F$, its closure denoted $b'$ is also bounded. We get that, for all $N\in \NN$,
$$\sum_{k\ge N} f_k(x)\subset \sum_{n\ge N}\tfrac 1{2^n}b'.$$
Hence, $\sum_k f_k(x)$ is \mca and so converges in $F$.


\end{proof}

\begin{prop}
\label{wcv_unif_pws}
Let $f:  E \rightarrow F$ be a \bornof function and let $f_k$ be $k$-monomials such that for every $ l \in F' $, $\sum_k l \circ f_k$  converges towards $ l \circ f$ uniformly on bounded sets of $E$. Then,  $f=\sum_k f_k$ is also a power series.
\end{prop}
\begin{proof}
Let $b$ be a bounded set and $b'$ be the absolutely convex and closed closure of $f(2b)$.
For any $l\in F'$, since $l\circ f$ is  a power series, it satisfies a Cauchy Inequality thanks to Proposition~\ref{Cauchy_ineq_pws} (notice that $l(b')$ is absolutely convex and that $(l\circ f)(2b)\subset l(b')$).
Therefore, for any $k\in\NN$, $(l\circ f_k) (2b) \subset l(b')$. By the Hahn-Banach Separation theorem and since $f_k$ is $k$-linear, we get that $f_k(b)\subset \tfrac{1}{2^k}b'$. Thus, $\sum_k f_k$ Mackey-converges uniformly to $f$ on bounded sets of $E$. Since Mackey-convergence entails convergence, we get that $f=\sum_k f_k$ is  a power series.

\end{proof}

The two last propositions helped us to infer strong convergence from weak convergence, the following will allow us to deduce uniform convergence from pointwise convergence.

\begin{prop} \label{simply_unif_cv} 
Let $\sum_k f_k: E \rightarrow F $ be a pointwise converging series of $k$-monomials. If the sum converges pointwise towards a \bornof function $f:E\rightarrow F$, then $f$ is a power series. 
\end{prop}
\begin{lemma}\label{lem:cvs-cvu}
Consider $E$ a Fr\'echet space and for every $ k \in \NN$ $f_k \in \mathcal{L}^k(E, \CC)$. Then $\sum f_k$ converges pointwise on $E$ if and only if it converges uniformly on every bounded set of $E$.
\end{lemma}

\begin{proof}
(see~\cite[I.7.14]{KriMi} for details) The reverse implication is straightforward. Let us prove the direct implication, and suppose $\sum f_k$ converges pointwise. Consider a $0$-neighbourhood $U$ in $E$. We want to show that $\{\tilde f_k(x_1,\dots, x_k)\ |\ k \in \NN, x_i \in U\}$ is bounded. If this is true, then $\sum f_k$ converges uniformly on $\lambda U$ for $\lambda <1 $, and thus on every bounded set of $E$. 

Since the domain of the $f_k$ is $\mathbb{C}$, their boundedness implies their continuity. Hence, the sets $$ A_{K,r} = \{ x \in E\ |\ \forall k\in\NN,\ |f_k(x^k)| \leq Kr^k\}$$
are closed. Moreover, they do recover $E$ by hypothesis. Then, by Baire property, there is an $A_{K,r}$ whose interior is nonempty. Consider $x_0 \in A_{K,r}$ and $V$  a neighbourhood of $0$ such that $ x_0 - U \subseteq A_{K,r}$. By the polarization formula (see Lemma~\ref{polarization_formula}), there is $\lambda>0$ such that for all $x \in U$ and $k\in\NN$, $ | f(x^k)| \leq K \lambda^k$ for some $\lambda > 0$. Then $\{ f_k(x^k)\ |\ k\in\NN \}$ is bounded on $\frac{U}{\lambda}$.
\end{proof}

\begin{proof}[Proof of Proposition~\ref{simply_unif_cv}]
Let us fix $l \in F'$. For every $x \in E$, $\sum_k l \circ f_k (x)$ converges towards $l \circ f(x)$ in $\CC$, and $l\circ f$ is bounded.
Let $b$ be a bounded set. Then, according to Lemma~\ref{lem:cvs-cvu}  which relates pointwise convergence and uniform convergence of power series on Banach spaces, the power series $\sum_k l \circ f_k (x)$ converges uniformly on $E_b$ (as it is a Banach space, see Proposition~\ref{E_b_mco}), hence on $b$. We have proved that $\sum_k f_k$ converges weakly uniformly on bounded sets. By Proposition~\ref{wcv_unif_pws}, we know that it converges (strongly) uniformly on bounded subsets. 
\end{proof}


\begin{prop}\label{comp_lin_pws}
Let $l$ be a linear \bornof function from $F$ to $G$ and $f= \sum_k f_k$ a power series from $E$ to $F$. Then $ l \circ f $ is a power series and $ l \circ f = \sum_n l \circ f_n$. 
\end{prop}

\begin{proof}
According to Proposition~\ref{MC_pws}, there is a sequence of scalars $(\lambda_n)$ decreasing towards $0$ and a bounded set $ B \subset \series(E,F)$ such that, for all $n$ :
$$ f- \sum_0^{n} f_k \in \lambda_n B.$$

Thus for every $n$, $l\circ f - \sum_0^{n} l \circ f_k \in \lambda_n l(B)$. Thus, applying this equation to every $x \in E$, we get that the partial sums of $\sum_k l \circ f_k(x)$ Mackey-converge towards $l \circ f(x)$. As $l \circ f$ is a bounded function, we have that $l \circ f$ is a power series thanks to Proposition~\ref{simply_unif_cv}.
\end{proof}

\subsection{A cartesian closed category}
\label{subsec:ccc}
\begin{defi}
Let us denote as $\series (E,F)$ the space of all power series between $E$ and $F$. We endow it with the topology of uniform convergence on bounded subsets of $E$. The bounded sets resulting from this topology are the equibounded sets of functions.
\end{defi}

Holomorphic maps, as defined in~\cite{KriMi} are in particular smooth~\cite[II.7.19.8]{KriMi}. Thus according to Proposition~\ref{powerseries_holo}, power series as defined here are smooth.

\begin{prop}
\label{bounded_incl_pws_smo}
We have a bounded inclusion of  $\series (E,F)$ into $\Cin(E,F)$.
\end{prop}

\begin{proof} 

Let $B$ be a bounded set in $\series (E,F)$. Let us prove that $B$ is bounded in $\Cin (E,F)$, \ie for every smooth curve $c \in \cC_E$, every bounded set $b \subset \RR$ and every $j \in \NN$, the following set is bounded in $F$:

$$\{ (f \circ c )^{(j)} (x)\ |\ f \in B,\ x \in b\}.$$

Let us fix $c\in \cC_E$ and $j\in\NN$. 

Let $C_j$ be the set made of $c$ and its derivatives  of order at most $j$. Since $c$ and its up to $j$th derivatives are smooth, they are bounded and send $b$ on a common absolutely convex bounded set $b'$ of $E$, \ie $C_j(b)\subset b'$. 

As a power series $f=\sum_n f_n$ converges uniformly on bounded sets of $E$, we can derivate under the sum. Thus, $(f \circ c )^{(j)} (x) = \sum_n  (f_n \circ c )^{(j)} (x)$.  It is possible to show by induction on $j$ that $(f_n \circ c )^{(j)}(x)=\left(\tilde{f_n}(c(\cdot),\dots,c(\cdot))\right)^{(j)}(x) = \sum_{l=1}^{j^n}\alpha_j^l  \tilde{f_n}(c^l_1(x),\dots,c^l_n(x))$ with $c_k^l\in C_j$ and $\alpha_j^l\le n^j$ an integer, where $\tilde f_n$ is the symmetric $n$-linear map from which $f_n$ results. Therefore, we have:
$$(f_n \circ c )^{(j)} (b) \subset n^j j^n\tilde{f_n}(b').$$ 
Now, let $b_E=4je\, b'$. According to Proposition~\ref{Cauchy_ineq_pws}, as $f(b_E)\subset B(b_E)$, we get:
$$f_n(b') \subset \tfrac 1{(4je)^n}B(b_E).$$ 
Thanks to the polarization formula (see Lemma~\ref{polarization_formula}), for any $x_1,\dots, x_n\in  b'$:
$$\tilde f_n(x_1,\dots,x_n)= \tfrac 1{n!} \sum_{\epsilon_1,\dots,\epsilon_n =0}^1 (-1)^{n-\sum_{j=1}^n \epsilon_j} f_n\left(\sum_{j=1}^n\epsilon_jx_j\right).$$
Then, for any $l\in F^\times$, we get:
\begin{eqnarray*}
  \left|l\circ\tilde f_n(x_1,\dots,x_n)\right|&\le& \frac 1{n!} \sum_{\epsilon_1,\dots,\epsilon_n =0}^1  \left|l\circ f_n\left(\sum_{j=1}^n\epsilon_jx_j\right)\right|\\
    &= & \frac 1{n!} \sum_{\epsilon_1,\dots,\epsilon_n =0}^1 (\textstyle\sum_{i=1}^n \epsilon_i)^n \left|l\circ f_n\left(\frac{\sum_{j=1}^n\epsilon_jx_j}{\sum_{j=1}^n\epsilon_j}\right)\right|\\
\end{eqnarray*}

Note that in the last sum, $\sum_j \epsilon_j$ can be supposed to be strictly positive, as when all $\epsilon_j$ equals $0$ then $ \left|l\circ f_n\left(\sum_{j=1}^n\epsilon_jx_j\right)\right| =0$. Now there is exactly $ \binom nj$ ways of having $\sum_{i=1}^n \epsilon_i =j$ :

$$\left|l\circ\tilde f_n(x_1,\dots,x_n)\right| \le  \frac 1{n!}  \sum_{j=0 }^n\binom nj j^n  \left|l\circ f_n\left(\frac{\sum_{j=1}^n\epsilon_jx_j}{\sum_{j=1}^n\epsilon_j}\right)\right|$$

However, by differentiating $n$ times the binom formula $(1+x)^n =  \sum_{j=0}^n \binom nj x^j$, one gets 
$$ \sum_{k=1}^n n \dots (n-k-1) x^{n-k}x^{k-1} = \sum_{j=0}^n j^n \binom nj x^j.$$

Taking $x =1$ thus implies $ \sum_{k=1}^n \frac{n !}{(n-k)!} 2^{n-k} =  \sum_{j=0 }^n\binom nj j^n $. We have then $$\frac 1{n!}  \sum_{j=0 }^n\binom nj j^n  = \sum_{k=0}^{n-1} \frac{1}{(k)!} 2^{k} \leq 2^n e.$$ 

Therefore :

\begin{eqnarray*}
\left|l\circ\tilde f_n(x_1,\dots,x_n)\right| &\le &   (2e)^n\left|l\circ f_n\left(\frac{\sum_{j=1}^n\epsilon_jx_j}{\sum_{j=1}^n\epsilon_j}\right)\right|
\\ &\le & (2e)^n \frac{1}{(4ie)^n}  |l\circ B(b_E)|\\
&  \le & \frac{1}{(2i)^n}|l\circ  B(b_E)|
\end{eqnarray*}

Thanks to Lemma~\ref{scalbound},  $b_F=n!(2i)^n \left\{\tilde f_n(x_1,\dots,x_n) \ \mid\ \forall f\in B,\  \forall x_1,\dots, x_n\in b'\right\}$ is bounded.

To conclude, for every $f \in B$,
$$(f_n \circ c )^{(i)} (b) \subset n^i\,i^n \tilde{f_n}(b') \subset \frac{n^i\,i^n}{n!(2i)^n}  b_F\subset \frac{n^i}{n!\,2^n}b_F, $$
so that,
$$(f \circ c )^{(i)} (b) \subset\sum_n (f_n\circ c)^{(i)}(b)\subset \sum_{n} \frac{n^{i}}{n!\,2^n}b_F. $$


\end{proof}

Let us note that any subset of $S(E,F)$ which is the restriction to $S(E,F)$ of a bounded set in $\Cin (E,F)$ is bounded. Indeed, according to \cite[4.4.7]{KriMi}, the bornology on $\Cin(E,F)$ is the coarsest one making all pointwise evaluations $ev_x : \Cin(E,F) \rightarrow F$ bounded. But when we artificially consider on $\Cin (E,F)$ the bornology of all uniformly bounded set, all pointwise evaluation are bounded. So this bornology is finer than the one resulting from the topology of $\Cin(E,F)$, that is bounded sets of $\Cin(E,F)$ are uniformly bounded. 

\begin{prop}
\label{prop:compl_series}
When $F$ is \mco so is $\series(E,F)$.
\end{prop}

\begin{proof}
Consider $(f_{\gamma})_{\gamma \in \Gamma}$ a \mca net in $\series(E,F)$.
There is a positive real net $(\lambda_{\gamma, \gamma'})_{\gamma,\gamma'\in \Gamma}$ converging towards $0$ and an equibounded set $B$ in $\series(E,F)$ such that

\begin{equation}\label{eq:CN}
f_{\gamma} - f_{\gamma'} \in \lambda_{\gamma, \gamma'} B.
\end{equation}

We can suppose \wlg that $B$ is absolutely convex and closed and that $B=\{f\ \mid\ \forall b\text{ bounded in }E,\ f(b)\subset B(b)\}$. For all $ x \in E$, $B(\{x\})$ is bounded in $F$, and $(f_{\gamma }(x))_{\gamma \in \Gamma}$ is a \mca net in $F$. Since $F$ is \mco, for each $x\in E$, $f_\gamma(x)$ converges towards $f(x)$ in F.

Let us show that $f : E \rightarrow F$ is a power series. Since $f_\gamma\in \series(E,F)$, we can write $f_{\gamma} = \sum_n f_{\gamma,n}$. Now, we fix $n \in \mathbb{N}$ and
prove that
$$ f_{\gamma,n} - f_{\gamma',n} \in \lambda_{\gamma, \gamma'} B.$$
From Equation~\eqref{eq:CN}, we have that, for any $b$ absolutely convex and bounded in $E$, $(\sum_n f_{\gamma,n}-f_{\gamma',n})(b)\in\lambda_{\gamma,\gamma'}B(b)$. Thus, by Proposition~\ref{Cauchy_ineq_pws}, for all $n\in\NN$, $(f_{\gamma,n}-f_{\gamma',n})(b)\in\lambda_{\gamma,\gamma'}B(b)$. We conclude by our assumption on the shape of $B$.

Then, $(f_{\gamma,n})_{\gamma \in \Gamma}$ is a \mca net in $\nlin (E,F)$, which is \mco according to Proposition~\ref{nlinmco}. Thus $(f_{\gamma,n})_{\gamma \in \Gamma}$ converges  in $\nlin(E,F)$ and we denote as $f_n$ its limit.

Let us show that $\sum_n f_n$ converges pointwise towards $f$. Let us  fix  $x \in E$ and $V$ an absolutely convex neighborhood of $0$ in $F$. We denote as $\DD x$ the  set $\{z x \ \mid z\in\CC, |z|<1\}$. 
We will show that each part of the following expression is small enough:
\begin{equation*}
  f(x)-\sum_{n<N}f_n(x) = \left( \lim_{\gamma'\to\infty} f_{\gamma'}(x)- f_\gamma(x)\right) +\left(f_\gamma(x)-\sum_{n<N}f_{\gamma,n}(x)\right) + \sum_{n<N}\left( f_{\gamma,n}(x)-f_n(x)\right)
\end{equation*}

Since $2\DD x$ is bounded, then so is $B(2\DD x)$ and there is $\mu>0$ such that $B(2 \DD x) \subset \mu V $. Let $\gamma_0 \in \Gamma$ be such that when $\gamma, \gamma'  \geq \gamma_0$, we have $ | \lambda_{\gamma, \gamma'} \mu |<1 $, and so $B( \DD x) \subset \lambda_{\gamma, \gamma'} B(2 \DD x)  \subset V$.  Then, 
\begin{equation*}
\forall\gamma',\gamma \geq \gamma_0,\ f_{\gamma'}(x) - f_{\gamma}(x)  \in \lambda_{\gamma, \gamma'} B( \DD x) 
\quad\text{and}\quad \lim_{\gamma'\to\infty} f_{\gamma'}(x)- f_\gamma(x)\in V.
\end{equation*}

By convergence, for $N \in \mathbb{N}$ big enough,
$$ f_{\gamma}(x)- \sum_{n < N} f_{\gamma,n}(x) \in V .$$


Moreover, for every $n$ we have 
$ f_{\gamma,n}(2x) - f_{\gamma',n}(2x) \in \lambda_{\gamma,\gamma'} B( 2 \DD x )$, and since they are $n$-monomials, $f_{\gamma,n}(2x) - f_{\gamma',n}(2x)= 2^n(f_{\gamma,n}(x) - f_{\gamma',n}(x)) $. Finally, by taking the limit $\gamma' \rightarrow \infty$, we get

\begin{equation*}
f_{\gamma,n}(x) - f_{n}(x) \in \frac{1}{2^n} \overline{V}
\end{equation*}

To sum up, $\sum_n f_n$ converges pointwise towards f, for $N$ big enough,
$$f(x)- \sum_{n < N} f_{n}(x) \in V + V +\left( \sum_{n < N} \frac{1}{2^n} \right) \overline V \subset 5 V. $$

Now, we apply Proposition~\ref{simply_unif_cv}, to show that $\sum
f_k$ does converge uniformly on bounded sets of $E$ towards $f$ and therefore
$f\in \series(E,F)$. It is sufficient to show that $f$ is bounded since we
have already shown the simple convergence. Let $b$ be an absolutely
convex and bounded set $b$ of $E$. Consider $\gamma \in \Gamma$. Then we get
$$ f(x)= f_\gamma(x) + (f(x) - f_\gamma(x)) = f_\gamma(x) + 
\lim_{\gamma'\rightarrow\infty} \sum_n (f_{\gamma', n}(x)-f_{\gamma,n}(x)).$$
 If $M$ is an upper bound of the net $(\lambda_{\gamma,\gamma'})$,
we get that $f(b) \subset f_{\gamma}(b) + M
\overline{B(b)}$.

\end{proof}

In order to prove that the composite of two power series is also a power series, we need to use Fubini's theorem and permute sums. We will have to embed our series in $\CC$ and to use Propositions~\ref{wcv_pointwise_pws} and~\ref{simply_unif_cv} that relates weak, strong, pointwise and uniform convergences.

\begin{theorem}
The composition of two power series is a power series.
\end{theorem}
\begin{proof}
Consider $ f = \sum_n f_n: E \rightarrow F$ and $g = \sum_k g_k: G \rightarrow E$ two power series. Let us show that $f \circ g: G \rightarrow F$ is a sum $\sum_m h_m$ of $m$-monomials converging uniformly on bounded sets of $G$. Let us use $\tilde{f_n}$ (resp. $\tilde{g_k}$) for the $n$-linear (resp. $k$-linear) function corresponding to $f_n$ (resp. $g_k$).

Because the series $\sum_k g_k$ Mackey-converges (see Proposition~\ref{MC_pws}), and because, for each $n \in \mathbb{N}$, $f_n$ is an $n$-monomial, we have:
 $$ \forall x \in G, \tilde{f_n}(g(x)) = \sum\limits_{k_1,\dots k_n \geq 0} f_n (g_{k_1} (x), \dots,g_{k_n} (x) ) .$$
Notice that $\tilde{f_n} (g_{k_1} (x), \dots,g_{k_n} (x) )$ is a $(k_1 + \dots + k_n)$-monomial in $x$. 

Let us write 
\begin{equation}\label{eq:hm}
 h_m: x \mapsto \sum\limits_{n \geq 0} \sum\limits_{k_1 + \dots + k_n = m \atop k_i \geq 0 }  \tilde{f_n} (g_{k_1} (x), \dots,g_{k_n} (x) )
\end{equation}

and show that $h_m$ is a well defined \bornof $m$-monomial such that $f \circ g = \sum_m h_m$. 

Let us consider $x \in G$ and fix $l \in F'$. The power series 
\begin{equation}\label{eq:lfg}
l \circ f \circ g  =  \sum\limits_{k_1,\dots,k_n \geq 0} l( \tilde{f_n} (g_{k_1} (x), \dots,g_{k_n} (x) ) )
\end{equation}
is convergent on $3 \DD x $, hence absolutely convergent on $ 2 \DD x $ where $\DD$ stands for the unit ball of $\CC$. Thus, we can permute coefficients in the converging sum above. Therefore, the general term $l\circ h_m(x)$ of the series $\sum_{m\ge 0} l \circ h_m (x)$, which is obtained from~\eqref{eq:lfg} by permuting indices of the sum, is also the sum of an absolutely converging series. By Proposition~\ref{wcv_pointwise_pws}, since for any $l\in F'$ and any $x\in G$, $l\circ h_m(x)$ is the limit of a converging sum in $\CC$, then for any $x\in E$, $h_m(x)$ is well-defined in $F$. Moreover, for any $l\in F'$, we have proved that $l\circ f\circ g(x)=\sum_{m\ge 0}l\circ h_m(x)=l\circ \sum_{m\ge 0}h_m(x)$, so by Hahn-Banach Separation theorem:
\begin{equation*}
  \forall x\in G,\ f\circ g(x)=\sum_m h_m(x).
\end{equation*}

Let $b$ be a bounded set in $G$. Since $g$ is \bornof, $g(2b)$ is a bounded set in $E$, and we set $b'$ its absolutely convex and closed closure which is also bounded. Let $b''$ be the absolutely and closed closure of the bounded set $f(2 g(2b))$ of $F$. Now, by Proposition~\ref{Cauchy_ineq_pws}, if $x \in b$, then $g_{k} (2x) \in b'$ and  $\tilde{f_n} (2g_{k_1} (2x), \dots,2g_{k_n} (2x) ) \in  b'' $. Since $g_{k_i}$ and $f_n$ are monomials,  for $x \in b$ we get $g_{k_i}(x) \in \frac{1}{2^{k_i}} b'$ and $\tilde{f_n} (g_{k_1} (x), \dots,g_{k_n} (x) ) \in \frac{1}{2^n} \frac{1}{2^{\sum k_i}} b''$. Since there is exactly ${m + n -1 \choose m}$ ways of choosing $n$ natural numbers whose sum is $m$, we get from formula \ref{eq:hm} :
$$h_m (x)  \in \frac{1}{2^m} \sum_n  {m + n -1 \choose m} \frac{1}{2^n} b'.$$
Moreover, we have:
 $${m + n -1 \choose m} \sim_{n \rightarrow \infty} \tfrac{n^m}{m!}.$$
Thus, $\sum_n {m + n -1 \choose m} \frac{1}{2^n}$ is absolutely converging. We have
$$h_m (b) \subset  \sum_n  {m + n -1 \choose m} \frac{1}{2^n} b'$$
so $h_m$ is \bornof. As it is a converging sum of $m$-monomials, $h_m$ is also an $m$-monomial.

We conclude that $f\circ g$ is a power series by Proposition~\ref{simply_unif_cv}, as $\sum_m h_m$ is a series of bounded $m$-monomials pointwise converging to $f\circ g$ which is also bounded.



\end{proof}

We can finally address the problem of cartesian closedeness, which is solved by getting back to the scalar case and by using Fubini's theorem. 

\begin{theorem}
\label{cart_closed}
When $E$, $F$ and $G$ are \mco spaces, then $$\series (E , \series (F,G) ) \simeq \series (E \times F , G).$$ 
\end{theorem}

\begin{proof}

Let us first notice that if the stated equality is true, then the topologies on these spaces are the same. Indeed, sending $B_1 \times B_2$ on a weak $0$-neighborhood $U$ is equivalent to sending $B_1$ on a function which will send $B_2$ on $U$. This will give us a homeomorphism, thus a \bornof isomorphism, between the two spaces. 

Let us define the two maps inverse of one another, as shown by direct computation:
\begin{equation*}
\phi: 
\left\lbrace
\begin{split}
\series( E \times F, G) & \rightarrow \series(E, \series(F,G)) \\
\sum_k f_k & \mapsto \left( x \mapsto \left( y \mapsto \sum\limits_n  \sum\limits_m \binom{n+m}{n} \tilde{f}_{n+m} (\underbrace{(x,0),\dots,(x,0)}\limits_{n \text{ times}},\overbrace{(0,y),\dots,(0,y)}\limits^{m \text{ times}}) \right) \right)
\end{split}
\right.,
\end{equation*}
and 
\begin{equation*}
\psi:
\left\lbrace
\begin{split}
 \series(E, \series(F,G)) & \rightarrow \series( E \times F, G) \\
 \sum\limits_n ( f_n: x \mapsto \sum\limits_m f^x_{n,m}) & \mapsto \left( (x,y) \mapsto \sum\limits_k \sum\limits_{n+m = k} f^x_{n,m} (y) \right) 
\end{split}
\right..
\end{equation*}
We need to show that they are well defined, linear and \bornof.
 The difficulty is in showing that their image is indeed made of power series. We will do it on $\psi$, the proof for $\phi$ using the same tools and being easier.

Consider a function $f \in \series(E, \series(F,G))$. Then $f$ can be written as $\sum_n ( f_n: x \mapsto \sum_m f^x_{n,m})$, each $f_n$ being a \bornof $n$-monomial from $E$ to $\series(F,G)$, and each $f^x_{n,m}$ being a \bornof $m$-monomial from $F$ to $G$. The function $ (x,y) \mapsto \sum_{n+m = k} f^x_{n,m} (y)$ is a \bornof $k$-monomial.

Let us fix $l \in G^{\times}$, $y\in F$ and define $\chi^y:\, \series(F, G) \to \CC,\: g\mapsto l\circ g(y)$. If $\cB$ is bounded in $\series(F,G)$, then $\cB(y)$ is bounded in $G$ and $\chi^y(g)$ is bounded in $\CC$, hence $\chi^y\in \series(F,G)^\times$.
Moreover, because $f$ is a power series, we know from Proposition~\ref{MC_pws} that its partial sums are Mackey-convergent and from Proposition~\ref{bounded_Mackey} that $\chi^y$ preserves Mackey-convergence. Thus, for any $x\in E$, we have that 
$$\sum_n \chi^{y} \left( \sum_m  f^x_{n,m} \right) = \sum_n \sum_m l \circ f^x_{n,m}(y).$$ 
In particular, let us fix $x$ and $y$, then $\sum_n \sum_m l \circ f^{2x}_{n,m}(2y)$ Mackey-converges in $\CC$. Therefore, $l \circ f^{2x}_{n,m}(2 y)=2^n2^ml \circ f^x_{n,m}(y)$ is the general term of a bounded double sequence and the radius of convergence of  the $\CC$-power series $\sum_n \sum_m l \circ f^x_{n,m}(y) z^{n+m} $ is at least $2$. Finally,  $\sum_n \sum_m l \circ f^x_{n,m}(y)$ converges absolutely in $\CC$. 
 Thanks to Fubini theorem, we know that  we can permute absolutely converging double series in $\mathbb{C}$. Then $\sum_k \sum_{n+m = k} l \circ f^x_{n,m} (y)$ converges and is equal to $\sum_n \sum_m l \circ f^x_{n,m}(y) $.   Thanks to Proposition~\ref{wcv_pointwise_pws}, for any $x\in E$ and $y\in F$,  $\psi(f)(x,y)\in  G$, that is $\psi(f)$ is pointwise convergent.

We now prove that $\psi (f)$ converges uniformly on bounded subsets of $E$. First, notice that $\psi(f)$ is \bornof. Indeed, $f$ is \bornof thanks to Proposition~\ref{pws_borno}, and $\psi (f)$ sends $B_1 \times B_2$ on $f(B_1)(B_2)$. Proposition~\ref{simply_unif_cv} states that a pointwise converging power series which converges towards a \bornof function converges uniformly on bounded subsets of its codomain. We conclude that $\psi (f) \in \series( E \times F, G)$. 

\end{proof}

\subsection{From $\Lin$ to $\Quant$}

So far, we have proven that the category $\Lin$ of \mco spaces and \bornof linear maps is symmetric monoidal closed and cartesian (see Section~\ref{sec:lin}). We have also proven that the category $\Quant$ of \mco spaces and smooth functions is cartesian closed (see Section~\ref{subsec:ccc}). We will now prove that there is a Linear-Non linear adjunction between $\Lin$ and $\Quant$ that comes from an exponential modality constructed exactly as in convenient spaces (see~\cite{BET10}\cite[5.1.1]{FroKri} and Section~\ref{subsec:DiLL}).

\begin{defi}
Let $E$ be a \mco space. For any $x\in E$, the Dirac delta distribution $\delta$ can be seen as a function on power series:
$$\delta:\left\lbrace
\begin{array}{rcl}
E&\rightarrow&\series(E,\CC)^\times\\
x&\mapsto&\delta_x:f\mapsto f(x)
\end{array}
\right.
$$
\end{defi}

\paragraph{Exponential modality.}

For any \mco space $E$, we construct a \mco space $!E$ from  $\delta(E)$ by applying the Mackey-completion procedure described in Proposition~\ref{mcompletion}.
\begin{defi}
Let us use $ ! E$ for the Mackey-completion of the linear span of $\delta (E) $ in $\series(E,\mathbb{C})^{\times}$ endowed with the topology of uniform convergence on bounded subsets of $\series(E,\mathbb{C})$.
\end{defi}

$\delta$ is clearly linear, and as it acts on \bornof functions (see Proposition \ref{pws_borno}) it is itself \bornof.

Thanks to Mackey-completion, in order to define a linear function on $!E$, it is sufficient to define it on $\delta_x$ for any $x\in E$. 
Let $f \in \linb(E,F)$ be a \bornof linear map. We define $!f:!E\to !F$ as the linear extension of:
\begin{equation*} 
!f: \left\lbrace 
\begin{split}
\delta (E) & \rightarrow !F \\
\delta_x & \mapsto \delta_{f(x)}
\end{split}
\right.
\end{equation*}
This function is linear by construction. Let us check that it is bounded. If  $B$ is an equibounded set in $\series(E,\CC)^{\times}$, then $\{\delta_{f(x)} \ |\  \delta_x \in B \} $ is equibounded. Indeed, if $B$ is bounded in $\series(E,\CC)$, then 
$$\{\delta_{f(x)}(B) \ |\  \delta_x \in B \} = \{ B(\{f(x)\})	 \ |\  \delta_x \in B \} = \{ \delta_x (B \circ f) \}	 \ |\  \delta_x \in B \}$$ is bounded, as $f$ bounded makes $B \circ f = \{ g \circ f \ |\  g \in B \}$ bounded. Hence $ !f $ is well defined, and is a bounded linear function. So we have indeed $ ! f \in \linb( !E,! F)$. 

\begin{defi}
  We write $ ! : \Lin \rightarrow \Lin$ for the functor sending a \mco space $E$ on $!E$, and a \bornof linear map $f\in\linb(E,F)$ on $!f\in\linb(!E,!F)$.
\end{defi}

\begin{prop}
  The functor $!$ is an exponential modality:
  \begin{itemize}
  \item $(!,\rho,\epsilon)$ is a comonad, with
\begin{equation*} 
\epsilon_E :
\left\lbrace
\begin{split}
!E & \rightarrow E\\
\delta_x & \mapsto x
\end{split}
\right.
\qquad\qquad\qquad
\rho_E : 
\left\lbrace
\begin{split}
!E & \rightarrow !!E\\
\delta_x & \mapsto \delta_{\delta_x}
\end{split}
\right. .
\end{equation*}
\item $!:(\Lin, \times,\top)\rightarrow (\Lin,\mctens, \unit)$ is a strong and symmetric monoidal functor, with
\begin{equation*} 
 m^0:
\left\lbrace
\begin{split}
1 & \rightarrow !\top=!\{0\}\\
1 & \mapsto \delta_0
\end{split}
\right.
\qquad\qquad\qquad
m^2_{E,F} : 
\left\lbrace
\begin{split}
!E\mctens !F&\rightarrow !(E\times F)\\
\delta_x \otimes \delta_y& \mapsto \delta_{(x,y)}
\end{split}
\right. .
\end{equation*}
\item the following diagram commute:
$$\xymatrix{
!E\mctens !F \ar[r]^{m^2_{E,F}} \ar[d]_{\rho_E\mctens\rho_F} & !(E\times F) \ar[r]^{\rho_{E\times F}} & !!(E\times F) \ar[d]^{!\langle!\pi_1,!\pi_2\rangle}\\
!!E\mctens !!F\ar[rr]_{m^2_{!E,!F}} &&!(!E\times !F)
}$$
  \end{itemize}


\end{prop}
\begin{proof}
  Notice that the natural transformations $\epsilon$, $\rho$ and $m^2$ are defined by linearity and Mackey-complete extension. Then, it is enough to check the diagrams for comonad and symmetric monoidality on Dirac delta distributions. The morphisms $m^0$ and $m^2_{E,F}$ are isomorphisms with inverse:
\begin{equation*} 
 (m^0)^{-1}:
\left\lbrace
\begin{split}
!\top=!\{0\} & \rightarrow  \unit\\
\delta_0 & \mapsto 1
\end{split}
\right.
\qquad\qquad\qquad
(m^2_{E,F})^{-1} : 
\left\lbrace
\begin{split}
!(E\times F)&\rightarrow !E\mctens !F\\
 \delta_{z}& \mapsto \delta_{\pi_1 z} \otimes \delta_{\pi_2 z}
\end{split}
\right. .
\end{equation*}
\end{proof}

\paragraph{Distributions.}

The distribution space $\series(E,\CC)^\times$ is equipped with a convolution product defined as follow. Notice that when restricted to $!E$, the convolution product can be obtained from the cartesian structure of $\Lin$ and from $m^2$.
\begin{prop}\label{prop:convo}
For any $D_1$ and $D_2$ in $\series(E,\CC)^\times$, the convolution $D_1\ast D_2$ is in $\series(E,\CC)^\times$ and acts on $f\in\series(E,\CC)$ as:
$$(D_1\conv D_2) f = D_1(x \mapsto (D_2 (y \mapsto f(x+y)))).$$
Moreover, if $D_1$ and $D_2$ are in $!E$, then $D_1\conv D_2$ is in $!E$.
\end{prop}
\begin{proof}
  Let $f\in \series(E,\CC)$ and $x\in E$. Since $(x,y)\mapsto x+y$ is linear and \bornof (and so a power series), the function $(x,y)\mapsto f(x+y)$ is a power series. Then, by cartesian closedness (Theorem~\ref{cart_closed}), $ x \mapsto ( y \mapsto f(x+y)) \in\series(E,\series(E,\CC))$. Since $D_2$ is \bornof and linear, we get by postcomposition that $x \mapsto D_2 ( y \mapsto f(x+y)) \in\series(E,\CC)$, thus we can apply $D_1$ to compute $(D_1\ast D_2) f$. Notice that $D_1\conv D_2$ is linear and \bornof since all the involved operations are both \bornof and linear.

Let $D_1$ and $D_2$ be in $!E$. Then the convolution operator $\conv$ is the morphism:
$$ \xymatrix@R=3pt@C=50pt{
  !E\mctens !E \ar[r]^{m^2_{E,E}} & !(E\times E) \ar[r]^{!((x,y)\mapsto x+y)} &!E\\
\delta_x \otimes \delta_y \ar@{|->}[r] & \delta_{(x,y)} \ar@{|->}[r]& \delta_{x+y}}
$$
Indeed, it is sufficient to prove it on Dirac delta distributions as they generate the \mco space $!E$.
\end{proof}

In general $\delta$ reflects the shape of the functions of its codomain (see~\cite{BET10} where $\delta$ is smooth). In Proposition~\ref{prop:deltaseries}, we show that $\delta$ is a power series by following the scheme introduced in~\cite{Ehr05}. First,  we focus on the maps $\theta_n:E\rightarrow \series(E,\CC)^\times$ that will be the components of the power series $\delta$.
\begin{defi}
Let $\theta_n: E\rightarrow \series(E,\CC)^\times$ be defined by induction on $n$ by:
\begin{equation*}
  \theta_0(x)=\delta_0,\qquad\qquad
  \theta_1(x)= \lim\limits_{t\to 0} \frac{\delta_{tx}-\delta_0}t,\qquad\qquad
  \forall n\in\mathbb N,\ \theta_{n+1}(x)= \theta_1(x)\conv\theta_n(x).
\end{equation*}
\end{defi}

\begin{prop}
\label{theta}
For any $n\in\mathbb N$, $\theta_n$ is a \bornof $n$-monomial from $E$ to $ !E$. Besides, for any $x\in E$ and $f=\sum_n f_n\in\series(E,\CC)$, we have
 $\theta_n(x)f=n!\,f_n(x)$.
\end{prop}

\begin{proof}
We prove this proposition by induction on $n\in\mathbb N$. Let $x\in E$ and $f=\sum_n f_n\in\series(E,\CC)$.

First, $\theta_0$ is constant, $\theta_0(x)=\delta_0$ in $!E$ and $\theta_0(x)f=f(0)=f_0(x)$.

Then, $\theta_1(x)(f)= \lim_{t\rightarrow 0}\frac{f(tx)-f(0)}t=f_1(x)$. Indeed, by Lemma~\ref{derivatives_pws},  the derivative of $c:z\in\CC\mapsto f(zx)$ at $0$ is $f_1(x)$. Besides, $\theta_1$ is linear as for $h\in\CC$, $\theta_1(x+h y)f =f_1(x+hy)=f_1(x)+hf_1(y)$ by linearity of $f_1$. Finally, notice that $t\mapsto \delta_{tx}$ is locally lipschitzian as for any $a\in\RR$ and $B\subset \series(E,\CC)$ equibounded, the set $\{\frac{f(tx)-f(0)}t\ |\ 0<t<a, f\in B\}\subset 2B(\{tx \ |\ 0<t<a\})$ is bounded. Thus, as proved in~\cite[Prop. I.1.7]{KriMi}, the net $\left(\frac{\delta_{tx}-\delta_0}t\right)_{t\in\RR}$ is Mackey-convergent and its limit $\theta_1(x)$ is in the \mco space $!E$.

Assume that $\theta_n(x)$ is in $!E$ and for any $g=\sum_n g_n$,  $\theta_n(x)g=n!\,g_n(x)$. Then thanks to Proposition~\ref{prop:convo}, $\theta_{n+1}(x)=\theta_1(x)\conv\theta_n(x)$ is in $!E$ and $$\theta_{n+1}(x)(f)=\theta_1(x)( y \mapsto \theta_n(x)( z \mapsto f(y+z))).$$ By induction hypothesis, 
$$\theta_n(x)(z\mapsto f(y+z))=n!\,\sum_{m\ge n} \binom mn \tilde f_m(\underbrace{y,\dots,y}\limits_{m-n},\underbrace{x,\dots,x}\limits_{n}),$$
where we denote by $\tilde f_m$ the symmetric $m$-linear \bornof map from which the $m$ monomial $f_m$ is constructed. So that, 
$$\theta_{n+1}(x)(f)=n!\,\binom {n+1}n \tilde f_{n+1}(x,\underbrace{x,\dots,x}\limits_{n})=(n+1)!\,f_{n+1}(x).$$
\end{proof}

As in~\cite{BET10}, the differential structure comes from the codereliction. Besides in this setting, this operator extracts the first coefficient of the power series.
\begin{prop}
The category $\Lin$ is equipped with a codereliction:
\begin{equation*}
\mathrm{coder_E}=\theta_1:
\left\lbrace
\begin{split}
E & \rightarrow !E \\
y & \mapsto \lim\limits_{t \rightarrow 0} \frac{\delta(ty) - \delta(0)}{t}
\end{split}
\right.
\end{equation*}
\end{prop}
\begin{proof}
  The strength and comonad diagrams of~\cite{Fiore}:
\[
\xymatrix@C=50pt{
E \mctens ! F \ar[r]^{{\sf coder}_E\mctens 1}
\ar[dr]_{1\mctens \epsilon_E}& 
! E \mctens! F \ar[r]^\phi&
! (E\mctens F)\\
& E\mctens F \ar[ur]_{{\sf coder}_{E\mctens F}}
}
\]
\[
\xymatrix@C=10pt{
& ! E \ar[dr]^\epsilon& \\
E \ar[ur]^{{\sf coder}_E} \ar[rr]_1&&E 
}
\qquad
\xymatrix@C=40pt{
E   \ar[d]|{\simeq}  \ar[r]^{{\sf   coder}_E}&   !  E\ar[r]^{\rho}
&!! E\\
E\mctens I \ar[r]_{{\sf coder}_E\mctens \nu}& 
! E \mctens ! E\ar[r]_{{\sf coder}\mctens \rho}
 & !! E \mctens !! \ar[u]_{\nabla}
E}
\]
 are shown exactly as in~\cite{BET10} since the actions of the involved natural transformations are defined similarly on the Dirac delta distributions.
\end{proof}

\begin{prop}\label{prop:deltaseries}
The map $\delta$ is a power series in $\series(E,\series(E,\CC)^\times)$:
$$\delta=\sum\limits_{n=0}^\infty \frac{\theta_n}{n!}.$$
\end{prop}

\begin{proof}
In order to show that $\delta$ is a power series, we apply Proposition~\ref{simply_unif_cv}. 

First, notice that $\delta$ is \bornof from $E$ to $\series(E,\CC)^\times$. Indeed, let $b$ be bounded in $E$, then $\delta(b)$ is equibounded in $\series(E,\CC)^\times$, since if $B$ is equibounded in $\series(E,\CC)$, $\delta(b)(B)=B(b)$ is bounded.

Now, let us prove that $\sum_{n=0}^\infty \frac{\theta_n}{n!}$ converges pointwise to $\delta$. Let $x\in E$, we need to prove that $\sum_{n=0}^\infty \frac{\theta_n(x)}{n!}$ converges to $\delta_x$ uniformly on bounded sets of $\series(E,\CC)$. We apply the Cauchy Inequality of Proposition~\ref{Cauchy_ineq_pws}. Let $b$ be absolutely convex such that $2x\in b$ and $B\in\series(E,\CC)$ be equibounded, then $B(b)$ is bounded in $\CC$, \ie there is $M$ such that $| f(y)| \leq M$ for every $ f \in B$ and $y\in b$. Thus, for any $f\in B$,  $|\frac{\theta_n(x)}{n!}(f)|=\frac 1{2^n}|f_n(2x)|\le \frac M{2^n}$ and the series $ \sum_n \frac{\theta_n(x)}{n!}$ converges uniformly on $B$. Its limit is $\delta_x$ as for every $f \in \series(E,\mathbb{C})$ and $x \in E$, we have $\delta_x(f) = f(x)= \sum\limits_{n=0}^\infty f_n (x) =\sum_{n=0}^\infty \frac{\theta_n (x)}{n!} (f) $. From this we conclude that pointwise, we have $\delta = \sum \theta_k$. 

As $\delta$ is bounded, Proposition \ref{simply_unif_cv} implies that the sum uniformly converges on bounded subsets of $E$. Thus $\delta$ is a power series.

\end{proof}

We just proved that we have a model of Intuitionist Linear Logic and thus, that the cokleisli category $\Lin_{!}$ is cartesian closed. We want now to show that the category $\Quant$ of \mco spaces and power series is isomorphic to $\Lin_!$, that is:
\begin{theorem}
\label{adjunctioncomp}
For every \mco space $E$ and $F$, we have the following bounded isomorphism $$ \series(E,F) \simeq \linb( ! E, F).$$ 
\end{theorem}

\begin{proof}
Consider $ f \in \series(E,F)$. Then define $\hat{f} : ! E \rightarrow F $ as $\hat{f} ( \delta_x) = f(x)$, extended linearly and Mackey-completed. We can define this function on $!E$ as $\hat{f}_{| \delta (E )}$ is \bornof : $\hat{f}_{| \delta (E )}^{-1} (U) = U_{\{ f\}, U} \cap \delta (E)$. By definition of the Mackey-completion of a lctvs, $ ! f$ is linear and \bornof. 

Now consider $ g \in \linb( ! E, F)$ and define $\check{g} : E \rightarrow F$ by $\check{g}(x) = g (\delta_x) = g \circ \delta$. As $g$ is bounded, we have by Proposition \ref{comp_lin_pws} that $\check{g} =  \sum_k \tfrac{1}{k!}g (\theta_k)$.
%
%

We check that $\hat{\check{g}} = g$, $\check{\hat{f}} = f $, that $ g \mapsto \check{g}$ and $f \mapsto \hat{f}$ are both linear and \bornof as $\delta$ is, and this induces a bounded isomorphism which is natural in $E$ and $F$ and so the wanted adjunction. 
\end{proof}

%
%
%
%
%
%
%
This concludes our construction of our denotational model of Linear Logic.
\begin{theorem}
\label{quantmodelILL}
The category $\Lin$, equipped with the comonad $ !$, is a quantitative model of intuitionist Linear Logic whose cokleisli category is $\Quant$, and a differential category.
\end{theorem}





\section{Quant is not *-autonomous}
One of the limits of the approach with bornologies is the extension to *-autonomous categories~\cite{Bar79}. Indeed, one could transform this model into a model of (classical) Differential Linear Logic by considering pairs $(E,\Ec)$ of Mackey-complete spaces, where $\Ec$ denotes the spaces of all bounded linear forms on $E$. This would be a construction alike the Chu construction.  

It is however difficult to have a more intrinsic approach. One could define a notion of b-reflexive space, as a space which equals its bounded bidual $\Ecc$. However, there is no handy Hahn-Banach theorem for bounded linear maps (see~\cite{Hog70}), and one cannot prove that the symmetric monoidal category of b-reflexive Mackey-complete spaces and \bornof maps is closed. 
Let us point out that this problem is not simpler with usual reflexive spaces, as the category of reflexive topological spaces and linear continuous maps is notoriously not closed. For example if we consider the bi-dimensional reflexive Hilbert space $l^2$, the space $\mathcal{B}(l^2)$ of bounded (equivalently continuous) endomorphisms in not reflexive (nor b-reflexive).

\section*{Conclusion}
This paper may be seen as a quantitative adaptation of~\cite{BET10}. It also brings a smooth and general point of view on quantitative semantics.
One can try to understand the computing meaning of this structure of power series, as some refinement to quantitative semantics. Indeed, many constructions of the present work relies on the Cauchy formula that power series satisfy. The same phenomenon happens in Girard's Coherent Banach spaces~\cite{Gir96}.

The next step now in understanding smooth models of Differential Linear Logic would be to go towards differentiation in manifolds. The first step in this direction would be to work on the logic underneath the theory of diffeology.

\bibliographystyle{alpha}
\bibliography{biblioMC}

\label{lastpage}
\end{document}